\documentclass[english,10pt,letterpaper,twocolumn,floatfix,allowfontchangeintitle,amsmath,amssymb,accepted=2023-04-18]{quantumarticle}
\pdfoutput=1
\usepackage[utf8]{inputenc}
\usepackage[english]{babel}
\usepackage[T1]{fontenc}
\usepackage{amsmath}
\usepackage{hyperref}
\usepackage{dsfont}
\usepackage[normalem]{ulem}
\usepackage{mathtools}
\usepackage[makeroom]{cancel}
\usepackage{bbm}
\usepackage{enumitem}
\usepackage{xcolor}
\usepackage{graphicx}
\usepackage{amsthm}
\usepackage{amssymb}             % AMS Math

\newtheorem*{theorem*}{Theorem}
\theoremstyle{definition}

\newtheorem*{principle*}{Principle}
\usepackage{algorithm}
\usepackage{algpseudocode}
\usepackage{comment}

\usepackage[normalem]{ulem}

\newcommand{\da}{^{\dagger}}

\newcommand{\id}{\mathbb{I}}
\newcommand{\ids}{\mathbbm{1}}
\newcommand{\tr}{\mathrm{Tr}} %old
\newcommand{\Tr}[1]{\mathrm{Tr}\left[ #1\right]}

\newcommand{\ketbra}[2]{\ket{#1} \!\! \bra{#2}}
\usepackage[normalem]{ulem}
\usepackage{braket}
\usepackage{xparse}
\ExplSyntaxOn
\NewDocumentCommand{\mref}{m}{\quinn_mref:n {#1}}
\seq_new:N \l_quinn_mref_seq
\cs_new:Npn \quinn_mref:n #1
{
	\seq_set_split:Nnn \l_quinn_mref_seq { , } { #1 }
	\seq_pop_right:NN \l_quinn_mref_seq \l_tmpa_tl
	( % print the left parenthesis
	\seq_map_inline:Nn \l_quinn_mref_seq
	{ \ref{##1},\nobreakspace } % print the first references
	\exp_args:NV \ref \l_tmpa_tl % print the last or only one
	) % print the right parenthesis
}
\ExplSyntaxOff

\definecolor{ms}{rgb}{0,.4,1}
\definecolor{js}{rgb}{0.578125,0.23828125,0.9609375}
\newcommand{\ms}[1]{#1}
\newcommand{\js}[1]{#1}

\newcommand{\sqrbra}[1]{\left[ #1\right]}
\newcommand{\de}{{\rm d}}
\newcommand{\J}{\mathbb{J}}
\theoremstyle{definition}
\newtheorem*{observation*}{Observation}
\newcommand{\cJ}{\mathcal{J}}
\newcommand{\R}{\mathcal{R}}
\newcommand{\norbra}[1]{\left( #1\right)}

\DeclareMathOperator*{\argmin}{arg\,min}

    \makeatletter
\def\@fnsymbol#1{\ensuremath{\ifcase#1\or \dagger\or *\or \ddagger\or
   \mathsection\or \mathparagraph\or \|\or **\or \dagger\dagger
   \or \ddagger\ddagger \else\@ctrerr\fi}}
    \makeatother

\usepackage[numbers,sort&compress]{natbib}
   
\begin{document}
%\raggedbottom
	\title{State retrieval beyond Bayes' retrodiction}
	\author{Jacopo Surace}\email{jacopo.surace@icfo.eu}
	\affiliation{ICFO - Institut de Ciencies Fotoniques, The Barcelona Institute of\\
		Science and Technology, Castelldefels (Barcelona), 08860, Spain}
	\author{Matteo Scandi}\email{matteo.scandi@icfo.eu}
	
	\affiliation{ICFO - Institut de Ciencies Fotoniques, The Barcelona Institute of\\
		Science and Technology, Castelldefels (Barcelona), 08860, Spain}
	%\date{\today}
	\begin{abstract}
	
				 In the context of irreversible dynamics, the meaning of the reverse of a physical evolution can be quite ambiguous. It is a standard choice to define the reverse process using Bayes' theorem, but, in general, this is not optimal with respect to the relative entropy of recovery. In this work we explore whether it is possible to characterise an optimal reverse map building from the concept of state retrieval maps. In doing so, we propose a set of  principles that state retrieval maps should satisfy. We find out that the Bayes inspired reverse is just one case in a whole class of possible choices, which can be optimised to give a map retrieving the initial state more precisely than the Bayes rule. Our analysis has the advantage of naturally extending to the quantum regime. In fact, we find a class of reverse transformations containing the Petz recovery map as a particular case, corroborating its interpretation as a quantum analogue of the Bayes retrieval.
				Finally, we present numerical evidence showing that by adding a single extra principle one can isolate for classical dynamics the usual reverse process derived from Bayes' theorem.
		%Once a state is prepared, one is faced with the task of minimising and correcting the deteriorating action of the external environment acting on it. Given the description of the dissipative dynamics so induced, a standard choice is to retrodict the initial state using Bayes' theorem. In this work we explore whether this choice is optimal. In doing so, we propose a set of natural axioms that state retrieval maps should satisfy. We find out that the Bayes retrodiction is just one case in a whole class of possible choices, which can be optimised to give a map retrieving the initial state more precisely than the Bayes rule. Our analysis mainly focuses on classical stochastic channels, but we also show how this axiomatic approach can be extended to the quantum regime. In this context, we find a class of reverse transformations containing the Petz recovery map as a particular case, corroborating its interpretation as quantum extension of the Bayes reverse.

	\end{abstract}
	\maketitle
	
	\section{Introduction }

	Reversible transformation of a physical system are bijective mapping between input and outputs. They are called reversible when a well defined notion of reverse operation exists, the latter of which involves the inversion of the direction of the element-wise mapping from the space of the outputs to the space of inputs. \js{Reversible quantum channels are unitary channels, while reversible classical stochastic processes are permutations.}
		
	Whenever the bijectivity between the space of inputs and outputs is lost, the standard definition of reverse operation no longer applies and one is forced to define a notion of generalised reversion.
	
	To this end, an illuminating approach is adopting a statistician's perspective and associating reverse processes with the process of retrodiction.
	It has been shown in~\cite{watanabe1955,watanabe1965,buscemi2021,aw2021,crooks2008quantum} that the common method for defining a generalised reverse map is analogous to the operation of retrodiction based on Bayes' theorem. In particular, considering the left-stochastic matrix $\Phi$ as the conditional probability $\varphi(i|j)=\Phi_{i,j}$ of obtaining the micro-state $i$ from the micro-state $j$, the Bayes inspired reverse map $\tilde{\Phi}_B$ is defined in coordinates as:
	\begin{equation}
		\label{eq:BayesInversion}
		(\tilde{\Phi}_B)_{i,j}=\frac{\Phi_{j, i}\,\pi_i}{(\Phi(\pi))_j},
	\end{equation}
	where $\pi$ is a fiducial state that is perfectly retrieved, called prior in Bayesian inference.
	
	Even though the choice of this specific reverse map can be thoroughly justified in the context of classical Bayesian inference ~\cite{jaynes2003,bernardo2009}, as we will see, it is just one of the many different reasonable reverse maps. Furthermore the notorious difficulty of extending the Bayes rule to quantum systems~\cite{asano2012,dezert2018,parzygnat2020,vanslette2017,warmuth2014,vanslette2018,holik2013,fuchs2009,Giffin07,Ali12,kostecki2014,Accardi78,Accardi82,Leifer06,coecke12}, %dezert2018, shafer85,culbertson14
	together with the partial arbitrariness of this choice, makes the characterisation of a quantum reverse maps even more questionable. In the quantum scenario, what can arguably be called the standard reverse map is the Petz recovery map~\cite{ohya1993}. This map has been introduced in relation to its properties in the context of the data processing inequality~\cite{petz1986,Petz:1988usv,Petz:2002eql,junge2018} but, again, it was shown in~\cite{buscemi2021,aw2021} that the Petz map can be regarded as one of the possible extensions of the Bayes inspired reverse map in the quantum context, and that it is a fundamental tool in deriving fluctuation theorems ~\cite{buscemi2021,aw2021}. 

In the following we tackle the problem of the arbitrariness in the choice of a generalised reverse map, introducing a definition of the class of state retrieval maps based on a set of physical desiderata. We will differentiate between state retrieval maps and reverse (or retrodiction) maps, considering reverse maps as state retrieval maps with the additional property of being involutive. We show that by choosing a maximisation principle we can single out a unique optimal map that outperforms the Bayes retrodiction in the task of state retrieval. 
The advantage of this construction is that, being based on a set of physical principles, it can be naturally extended to the quantum regime, partly overcoming the difficulty in directly extending the Bayes rule. Even in the quantum case, using an analogous maximisation principle, we show that in the considered example the optimal map outperforms the Petz map.

Finally, we present  numerical evidence suggesting that \ms{by adding a desideratum, namely that reversing an evolution should be involutive (i.e., the reverse of the reverse is the forward map)}, it is possible to single out a unique map coinciding with the  Bayes inspired retrieval map \ms{for classical stochastic maps}. 

\subsection{Rationale}

Before entering the technical details, it is important to explain the intuition that guides our construction. The main aim of this work is to define a physical process that can recover the initial conditions of a dissipative dynamics $\Phi$ as accurately as possible. For unitaries there is a unambiguous choice given by $\Phi^{-1}$, but this is not well defined for general processes, as the inverse of a dissipative evolution is unphysical. For this reason, we put forward a construction of the main desiderata for a generic inverse $\tilde{\Phi}$. 

First, it should be noticed that any physically realisable process with domain equal to codomain has at least one fixed point. This gives us the freedom to encode any additional information on the initial conditions into the fixed point of $\tilde{\Phi}\Phi$, i.e., one can always choose without loss of any generality a state $\pi$ that will be perfectly recovered. 

Moreover, we impose that the statistics of $\tilde{\Phi}\Phi$ are as time symmetric as possible. This requirement takes the form of detailed balance condition on the transition rates, which is the canonical method of enforcing time symmetry in dissipative dynamics. It is an easy exercise to prove that the only evolution which is detailed balance with respect to every state is the identity map, which in this case corresponds to the undesirable  $\tilde{\Phi}=\Phi^{-1}$. Hence, in order not to lose generality, we require detailed balance with respect to a single state. It is a standard result that if a map satisfies detailed balance with respect to a state, then this is also a fixed point of the evolution. For this reason, choosing any other state than the $\pi$ defined above would introduce some extra information about the initial conditions, namely that this additional state should perfectly retrieved. For this reason, we require the same state $\pi$ to be the one with respect to which $\tilde{\Phi}\Phi$ is time symmetric.

The last requirement is geometrical in nature. It should be noticed that any rotation in the image of $\tilde{\Phi}\Phi$ can only decrease the quality of the retrieval in a trivial manner: in fact, any rotation of the image can be undone by simply rotating back at the end of the protocol. For this reason, without loss of generality one can assume that the image of $\tilde{\Phi}\Phi$ has the same orientation as the original space of states. This intuitive argument is mathematically encoded by principle~\ref{ax:5}.

We show that the family of possible retrieval maps $\tilde{\Phi}$ satisfying the principles above is actually a convex space. To benchmark our construction, then, we also verify that the most common choice of state retrieval, i.e., the one coming from Bayes' retrodiction, is indeed contained in the set that we defined. Still it should be kept in mind that this work is not primarily interested in reconstructing the Bayes inversion (a topic covered in Sec.~\ref{sec:Bayes-rev}), but rather in exploring the intuitive definition of retrieval maps. 

Lastly, the principle we choose to single out the optimal map is also geometrical. In particular, we start from the consideration that any physical map compresses the space of states. This implies the existence of some states outside of the image of $\tilde{\Phi}\Phi$  that simply can't be retrieved. Following this intuition, we assess the quality of a retrieval map $\tilde{\Phi}$ by  how big the volume of the image of $\tilde{\Phi}\Phi$ is, or, dually, by how small the volume of inaccessible states is. Then, the optimal map should be the one maximising this volume (see Fig.~\ref{fig:Traslazione-Composed} for an illustrative example). 

The discussion above leaves open the question of how to measure volumes in the phase space. Our choice is to look at the determinant of $\tilde{\Phi}\Phi$. This is motivated by the following two reasons: first, since $\tilde{\Phi}\Phi$ is a linear map, a standard result from linear algebra tells us that the Euclidean volume of its image is given by the determinant, so our choice aligns with the canonical treatment. Secondly, the determinant can be efficiently optimised through convex optimisation. This last property is particularly desirable when taking into consideration applications to concrete physical problems. 

	\section{Characterisation of state retrieval}
	
	%In this section we postulate the desiderata that a state retrieval should satisfy. The resulting class of maps contains the Bayes' rule as a particular case. We also show how a maximisation principle can be used to efficiently compute what we argue is the optimal state retrieval map.
	
	\subsection{General requirements}
	
	In order to characterise which maps can be useful as state retrieval, we put forward some minimal desiderata that they should satisfy. 
	
	\ms{Suppose one wants to revert a map $\Phi$ given some information on the initial conditions of the system encoded by a fiducial state $\pi$, called the \emph{prior}. The first principle we define is that the state retrieval should be physically implementable, which mathematically corresponds to:}
		
	\begin{enumerate}
		\item The state retrieval map $\tilde \Phi$ is described by a left stochastic matrix; \label{ax:1}
	\end{enumerate}
	
	\ms{Notice that this requirement prevents one from setting $\tilde{\Phi}=\Phi^{-1}$: in fact, for dissipative evolutions $\Phi^{-1}$ is not a stochastic matrix, so it cannot be physically realised, as it would send general states to something that is not a probability vector. Still, for reversible transformations (i.e., for permutations) $\Phi^{-1}$ is indeed a stochastic map which perfectly recovers the input of $\Phi$. Since this choice is obviously optimal, we also require that:}
	\begin{enumerate}[resume]
	\item \ms{If the map $\Phi^{-1}$ exists and it is left stochastic, the state retrieval $\tilde \Phi$ coincides with it,} \label{ax:2}
	\end{enumerate}
	\ms{that is, whenever it is possible, we should set $\tilde{\Phi}=\Phi^{-1}$.
	
	Notice that since both $\Phi$ and $\tilde \Phi$ are stochastic, this also holds for their composition $\tilde{\Phi}\Phi$. This implies that the composite map has at least one probability vector associated to the unitary eigenvalue, corresponding to the state that is perfectly recovered by the retrieval map. Thanks to this fact, we can encode the information about the initial conditions in it, that is the prior state $\pi$ should correspond to an eigenvector of the composite evolution with eigenvalue one. Hence, the third requirement is:}
	\begin{enumerate}[resume]
		\item The prior state is one of the perfectly retrieved states: $\tilde\Phi (\Phi(\pi))  =  \pi$.  \label{ax:3}
	\end{enumerate}

	%The intuition behind this requirement is the following: assume that the system, initially unperturbed, starts evolving under the dynamics $\Phi$; at the end of the transformation, one observes the state $\Phi(\pi)$ and has to guess the initial distribution. Even if many states might be mapped into $\Phi(\pi)$, the fact that a system generically spends most of its time at equilibrium makes the state $\pi$ the most plausible choice. 
	%In Bayesian inference the prior $\pi$ together, with the direct map $\Phi$, constitutes the total amount of knowledge on the system before any measurement.
	%Another interpretation is that if a state is the exact evolution of our initial belief, then the reverse maps should recover exactly our initial belief without removing or adding any new information. Hence, any proper reverse map should take into account this observation.
	
	Finally, \ms{we do not only require that $\pi$ is one equilibrium state of the dynamics $\tilde\Phi \Phi$, but also that this is detailed balanced with respect to it. This can be expressed in coordinates as 
		\begin{align}
			(\tilde\Phi\Phi)_{j,i} \, \pi_i  = (\tilde\Phi\Phi)_{i,j} \, \pi_j,\label{app:eq:detailed balance1}
		\end{align}
	and it corresponds to the requirement of time symmetric dynamics in  the associated Markov chain. This request can be interpreted as follows: since $\tilde \Phi \Phi$ corresponds to an evolution forth-and-back, its statistics should not distinguish between the two directions of time. By this we mean that the probability of measuring the microstate $i$ at the beginning and evolving to $j$ should be the same as the one of first measuring $j$ and ending up in the $i$-th state. Unfortunately, we cannot impose such a strong requirement for all states, as it would lead to the unphysical $\Phi^{-1}$. For this reason, we limit ourselves to imposing time symmetry in the rates of the dynamics with respect to the prior state, as expressed in Eq.~\eqref{app:eq:detailed balance1}:
	}
	
	%\ms{It is useful to make explicit the constraints imposed by condition~(\ref{ax:3}) and~(\ref{ax:4}) on $\tilde\Phi\Phi$. First, notice that the probability $p(i, j)$ defined in condition~(\ref{ax:4}) can be rewritten as:
		%\begin{align}
			%p(i,j) = (\tilde\Phi\Phi)_{j,i} \, \pi_i,
		%\end{align}
		%where we used the definition of prior state from requirement~(\ref{ax:3}). By imposing the symmetry in the arguments of $p(i, j)$, one obtains:
		%which can be rewritten in matrix form as:
		%\begin{align}
			%(\tilde\Phi\,\Phi )\,J_\pi= J_\pi \, (\tilde\Phi\,\Phi)^T .\label{app:eq:detailed balance}
		%\end{align}
		%One can recognise in Eq.~{(\ref{app:eq:detailed balance1}-\ref{app:eq:detailed balance})} the definition of detailed balance condition for the transitions of $\tilde\Phi\Phi$, corresponding to a time symmetric dynamics in the corresponding Markov chain.}
	
	%consider the following scenario: by starting from the perfectly recoverable state $\pi$, one first measures the microstate $i$, and evolves the state according to $\Phi$; at this point, the retrieval map $\tilde\Phi$ is applied and the micro-state $j$ is measured. Since $i$ and $j$ live on the same space, it is unproblematic to define a joint distribution $p(i,j)$. Notice that exchanging $i$ with $j$ corresponds to a time inversion. Requiring the probability $p(i,j)$ to be symmetric in its indices corresponds then to the natural requirement that the error probability only depends on the micro-state $i$ and $j$ and not on their order. For this reason, we require that:
	\begin{enumerate}[resume]
		\item \ms{The evolution $\tilde\Phi \Phi$ satisfies detailed balance with respect to $\pi$. } \label{ax:4}
	\end{enumerate}
	%This requirement can be seen as the analogous of requirement \ref{ax:2} for reversible maps, in the sense that it encodes what one means by state retrieval in the general case. 
	
	In order to explore which maps can be considered as possible candidates for a state retrieval, it is first useful to introduce a particular parametrisation of stochastic maps. This is the subject of the next section.
	\color{black}
	
	\subsection{Parametrisation of stochastic maps with a given transition}
	
	Consider the family of stochastic maps $\Psi$ with fixed transition $\Psi(\pi)=\sigma$, where both $\pi$ and $\sigma$ are probability vectors with strictly positive entries\footnote{As a standard approach, in the classical case we are not going to consider the case of probability vectors with zero entries as well as in the quantum case we are not going to consider rank-deficient density matrices.  These are special cases that in general inference studies are treated separately, taking care of the possible zeros appearing at the denominators.  In the field of Bayesian inference, for example, techniques used to deal with these scenarios are often referred as techniques to solve the \textit{zero frequency problem}. Moreover, if we assume that all the states are defined on a space of fixed dimension, since zero frequency vectors are always $\varepsilon$-close to a full rank one, one could also argue that given any finite precision in the experiment it is impossible to certify them. Thus, without loss of generality this pathological case can be neglected.}, and of the same dimension. These maps can be rewritten as:
	\begin{align}\label{eq:parametrisation}
		\Psi = \Lambda^\Psi \cJ_{\pi}^{-1},
	\end{align}
	where $\cJ_{\pi}$ is a diagonal matrix with entries $(\cJ_{\pi})_{i,i} := (\pi)_i$, and $\Lambda^\Psi$ is implicitly defined by the equation $\Lambda^\Psi:= \Psi \cJ_{\pi}$. This matrix satisfies the following two conditions. From the request that $\Psi$ is stochastic one can deduce that:
	\begin{align}
		\sum_{i} \,\Lambda^\Psi_{i,j} = \sum_{i} \,\Psi_{i,j}\,(\cJ_{\pi})_{j,j} = \pi_j.\label{eq:rowCond}
	\end{align}
	Moreover, since the transition $\Psi(\pi)=\sigma$ is specified,  $\Lambda^\Psi$ also satisfies:
	\begin{align}
		\sum_{j} \Lambda^\Psi_{i,j} = \sum_{j} \,\Lambda^\Psi_{i,j} \,(\cJ_{\pi}^{-1})_{j,j} (\pi)_j = \sigma_i .\label{eq:colCond}
	\end{align}
	This means that any stochastic map $\Psi$ with fixed transition $\Psi(\pi)=\sigma$ is uniquely identified by an element $\Lambda^\Psi$ of $\mathcal{U}(\sigma,\pi)$, the space of matrices with non-negative entries, with columns summing to $\sigma$ and rows summing to $\pi$. Interestingly, $\mathcal{U}(\sigma, \pi)$ is a convex polytope with finite number of vertices, denoted by $V_{\sigma|\pi}^{(k)}$~\cite{Jurkat1967} \ms{and indexed by $k$}. Moreover, since the matrix transpose exchanges the role of Eq.~\eqref{eq:rowCond} and Eq.~\eqref{eq:colCond}, the vertices of $\mathcal{U}(\sigma, \pi)$ and the one of $\mathcal{U}(\pi, \sigma)$ are in a one-to-one correspondence through the transformation $(V_{\sigma|\pi}^{(k)})^T = V_{\pi|\sigma}^{(k)} $.
	
	Putting everything together, we can then parametrise the matrix $\Psi$ as:
	\begin{align}
		\Psi = \sum_k\, \lambda_k^{(\Psi)} \,V_{\sigma|\pi}^{(k)}\cJ_{\pi}^{-1},\label{eq:generalParametrisation}
	\end{align}
	where $\{\lambda_k^{(\Psi)}\}$ are positive coefficients summing up to one.  
	\js{We note that the convex polytope  $\mathcal{U}(\sigma, \pi)$  is not in general a simplex. Thus,
an arbitrary element inside it can be parametrised by more than one convex combination of the vertices.
	Nevertheless, }this parametrisation gives a way to uniquely identify a map through a set of coefficients vector $\{\lambda^{(\Psi)}\}$ and the ordered pair of states $(\sigma,\pi)$ of the fixed transition.

	\subsection{Parametrisation of state retrieval maps}
	%It is clear that principles \mref{ax:1,ax:3} individuate the family of stochastic maps $\tilde{\Phi}$ such that $\tilde{\Phi}(\Phi(\pi))=\pi$, that is the family of stochastic maps with fixed transition $\Phi\pi\to\pi$. Thus,  in this context, the problem of finding the optimal map $\tilde{\Phi}_O$ for state retrieval corresponds in finding the set of coefficients $\{\lambda_k^{(\tilde{\Phi}_O)}\}$ such that,
	The parametrisation just presented can be used to easily enumerate all the possible retrieval maps. First, it should be noticed that the transformation $\Phi$ maps the prior state $\pi$ into $\Phi(\pi)$, meaning that it can be characterised by the vector of scalar coefficients $\{\lambda_k^{({\Phi})}\}$  in the following way 
	\begin{align}
		\Phi = \sum_k \,\lambda^{(\Phi)}_k\,V_{\Phi\pi|\pi}^{(k)} \cJ_\pi^{-1} ,\label{eq:classicalMap}
	\end{align}
	where $V_{\Phi\pi|\pi}^{(k)}$ are vertices of $\mathcal{U}(\Phi\pi,\pi)$. In the same spirit, since requirements~(\ref{ax:2},~\ref{ax:3}) impose that the retrieval map $\tilde{\Phi}$ is a left stochastic matrix with the  fixed transition $\tilde{\Phi}(\Phi(\pi))=\pi$, one can parametrise it as
	\begin{equation}
		\tilde{\Phi} =\sum_k \,\lambda^{(\tilde{\Phi})}_k\,V_{\pi|\Phi\pi}^{(k)} \cJ_{\Phi\pi}^{-1},
	\end{equation}
	where, in this case, $V_{\Phi\pi|\pi}^{(k)}$ are the vertices of $\mathcal{U}(\pi,\Phi\pi)$.  Thanks to the relation between  $\mathcal{U\textsl{}}(\Phi\pi, \pi)$ and $\mathcal{U}(\pi, \Phi\pi)$ the vertices in the two cases are connected by the transposition $(V_{\Phi\pi|\pi}^{(k)})^T = V_{\pi|\Phi\pi}^{(k)} $. For this reason we can focus solely on the coefficients vector, and associate to each state retrieval a transformation $\R$ that maps the coefficients vector $\{\lambda_k^{({\Phi})}\}$ to the coefficients vector $\{\lambda_k^{({\tilde\Phi})}\}$\footnote{\ms{To be more precise, a state retrieval map which sends $\Phi$ to $\tilde \Phi$ is defined on the quotient space $U(\Phi\pi,\pi)$, whose elements are equivalence classes of coefficients $[\{\lambda_k^{({\Phi})}\}]$, defined by the relation that two points are part of the same  equivalence class if they induce the same map on $\mathcal{U}(\Phi\pi,\pi)$ (and similarly for the image space $U(\pi,\Phi\pi)$). When passing to the original space of probability distributions $\{\lambda_k^{({\Phi})}\}$, the relation between state retrieval maps and the corresponding $\mathcal{R}$ is no longer one-to-one in general, but rather one-to-many. In particular, $\mathcal{R}$ should satisfy the implicit request of having a well-defined projection on the quotient space, namely, the corresponding state retrieval.}}.
	
	In the following sections we explore two possibilities for $\R$, one associated with the Bayes inspired reverse, the other with what we call the optimal state retrieval.
	
	\begin{figure*}[t!]
		\centering
		\includegraphics[width=1.\linewidth]{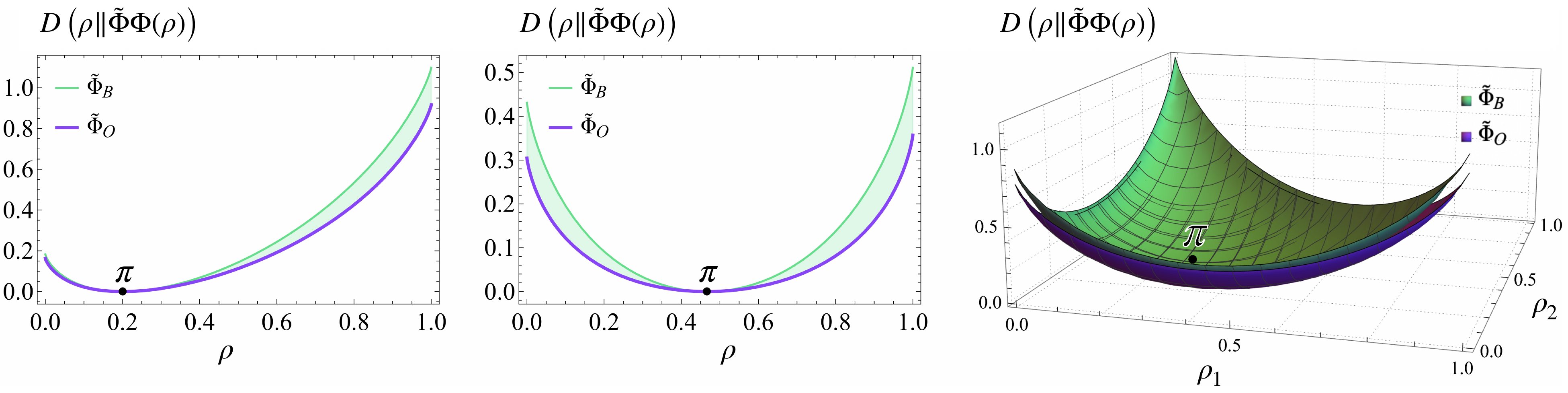}
		\caption{Relative entropy between a distribution and its evolution forwards and backwards. In  the first two plots we consider probability vectors $\rho=[\rho,1-\rho]$ of  length two, while in the third plot we consider probability vectors $\rho=[\rho_1,\rho_2,1-\rho_1-\rho_2]$ of length three. In all of the plots the map $\Phi$ and the prior distribution $\pi$ are chosen at random. As it can be seen, the optimal map $\tilde\Phi_O$ outperforms the Bayes retrodiction $\tilde\Phi_B$ in retrieving the original distribution in the whole space.}
		\label{fig:figcomposed}
	\end{figure*}
	
	\subsection{Bayes inspired reverse}\label{sec:bayesreverse}
	The Bayes inspired reverse defined in Eq.~\eqref{eq:BayesInversion} satisfies the desiderata~(\ref{ax:1}-\ref{ax:4}), so it is a legitimate state retrieval. Moreover, 
	it is a surprising fact that it corresponds to a particularly simple transformation of the coefficient vector $\{\lambda_k^{({\Phi})}\}$. In fact, by rewriting Eq.~\eqref{eq:BayesInversion} in matrix form one obtains:
	\begin{align}\label{eq:bayesDerivation}
		\tilde\Phi_B& = \cJ_\pi\,\Phi^T \cJ_{\Phi\pi}^{-1} = \sum_k \,\lambda^{(\Phi)}_k\,(V_{\Phi\pi|\pi}^{(k)})^T \cJ_{\Phi\pi}^{-1}=\\
		&=\sum_k \,\lambda^{(\Phi)}_k\,V_{\pi|\Phi\pi}^{(k)} \cJ_{\Phi\pi}^{-1}. \label{eq:baDecomposition}
	\end{align}
	Hence, in this case $\R$ corresponds to the identity transformation $\R[\{\lambda_k^{({\Phi})}\}]_i := \lambda_i^{({\Phi})}$.

	\subsection{Optimal state retrieval}
	
	Principles~(\ref{ax:1}-\ref{ax:4}) do not select a unique retrieval map, but rather a whole family of transformations. After specifying one more requirement, we provide a maximisation principle that singles out a unique optimal state retrieval map $\tilde{\Phi}_O$.
	
	To this end, consider a stochastic map from a space into itself. These types of maps are contracting: the volume of their image will be smaller than the one of their domain. The composite transformation $\tilde\Phi\Phi$ falls into this class. Intuitively, it can be argued that the optimal state retrieval should maximise the volume of the image of $\tilde\Phi\Phi$.
	
	Similar considerations lead us to impose one more requirement on $\tilde\Phi$. Notice, in fact, that any negative or complex eigenvalue in the spectrum of $\tilde\Phi\Phi$ corresponds to a reflection or a rotation of the domain, which would increase the statistical distance between a state and its evolved version. For this reason, we impose the principle:
	\begin{enumerate}[resume]
		%\item A map satisfying axioms~$(1$-$4)$ is an reverse if all the eigenvalues of $\tilde\Phi\,\Phi$ are non-negative.  \label{ax:5}
		\item The map $\tilde{\Phi}$ is a state retrieval map if all the eigenvalues of $\tilde\Phi\,\Phi$ are non-negative.  \label{ax:5}
	\end{enumerate}
	%This implies that the Markov chain defined by the transition probability matrix $\tilde\Phi\Phi$ is reversible, i.e., it gives the same dynamics forwards and backwards in time. \ms{Poi ci penso bene.}
	
	It should be noticed that \ms{the} Bayes inspired reverse still falls in this class of transformations. In fact, one can rewrite $\tilde \Phi_B\Phi $ as:
	\begin{align}
		&\tilde \Phi_B \Phi = \cJ_\pi\,\Phi^T \,\cJ_{\Phi\pi}^{-1} \,\Phi =\\
		&= \cJ_\pi\sqrbra{ (\cJ_{\Phi\pi}^{-1/2} \,\Phi )^T\, (\cJ_{\Phi\pi}^{-1/2} \,\Phi\,) }.
	\end{align}
	Both the matrix in the square parenthesis and $\cJ_\pi$ are positive semidefinite. The product of two positive semidefinite matrices has positive spectrum, so \ms{the} Bayes inspired reverse satisfies principle~(\ref{ax:5})\ms{.}
	
	Despite the fact that this requirement might appear to be a strong restriction on the class of possible maps, it is still not sufficient to single out a unique transformation. For this reason, we define the optimal retrieval map by the following:
	
	\begin{principle*}
		The optimal retrieval map is defined to be the $\tilde{\Phi}_O$ that maximises the determinant of $\tilde{\Phi}_O\,\Phi$ under the constraints (\ref{ax:1}-\ref{ax:5}).
	\end{principle*}
	In this case then the transformation $\R$ assigns to $\{\lambda_k^{({\Phi})}\}$ the vector $\{\lambda_k^{(\tilde\Phi_O)}\}$ corresponding to the solution of the maximisation problem
	\begin{align}
		\max_{\substack{\tilde \Phi \,{\rm state\, retrieval}}} \;\; \det \tilde \Phi \Phi .
	\end{align}
	In section~\ref{sec:optimalStateRetrieval} we provide an efficient algorithm to construct $\R$, which also proves the uniqueness of the solution. Before passing to that, we provide in the next section an analytic justification to the principle just presented.
	
	%The use of the determinant is justified by its interpretation as a measure of the volume of the image of a linear map. Reconnecting to the discussion above, we expect the optimal state retrieval to be able to preserve as much volume in the phase space as possible. This intuition, together with the discussion in the next section, justifies the principle.

	\subsection{Quality of the retrieval}
	
	Beyond the intuitive necessity of having the image of the retrieval map as big as possible,  the principle of the determinant maximisation can be justified more rigorously. We present here \ms{some} arguments \ms{explaining} why the optimal retrieval map should be the one that maximises the determinant.
	
	Consider, as a first example, the average relative entropy between the original distribution and the one evolved forward and back:
	\begin{align}\label{eq:averageRelativeEntropy}
		\int_{\mathcal{S}} \de \rho \; \ms{D(\rho || \tilde\Phi\Phi ( \rho))} = \int_{\mathcal{S}} \de \rho \; \rho\cdot (\log \rho - \log \tilde\Phi\Phi ( \rho)),
	\end{align}
	where we indicate by $\mathcal{S}$ the space of states. Thanks to the properties of the relative entropy, this average is always non-negative, while it is zero if and only if $\tilde\Phi\Phi ( \rho) = \rho$ for every $\rho$, implying that $\tilde\Phi\Phi \equiv \id$. We show in Appendix~\ref{app:relativeEntropyEstimate} that \ms{for invertible $\tilde\Phi\Phi $ (notice that non-invertible maps only have measure zero, as they are not stable under any arbitrarily small perturbations)} this quantity satisfies the inequality:
	\begin{align}\label{eq:dEstimate}
		0\leq	\int_{\mathcal{S}} \de \rho \; \ms{D(\rho || \tilde\Phi\Phi ( \rho))} \leq\frac{K}{|\det \tilde{\Phi}\Phi|}-\langle S(\rho)\rangle,\
	\end{align}
	where $K$ is a numerical constant independent of $\tilde \Phi$ and $\langle S(\rho)\rangle$ is the average Shannon entropy. This chain of inequalities gives an idea about why maximising the determinant also minimises the average relative entropy between the initial state and the retrieved one.
	
	A more precise argument follows from the observation that in order to optimise the quality of the retrieval we have to make $\tilde{\Phi}\Phi$ as similar as possible to the identity transformation. Since both $\tilde{\Phi}\Phi$ and $\id$ are positive semidefinite matrices, the relative entropy between the two is well defined and takes the form:
	\begin{align}
		D(\id||\tilde{\Phi}\Phi) &= \Tr{\id\, (\log \id - \log\tilde{\Phi}\Phi ) } =\\&
		= -\Tr{ \log\tilde{\Phi}\Phi } = \log \det (\tilde{\Phi}\Phi )^{-1},\label{eq:relEntropyIdFB}
	\end{align}
	where we used the well known matrix identity $\Tr{\log A} = \log \det A$. Minimising this relative entropy is then equivalent to the maximisation of the determinant of $\tilde{\Phi}\Phi$. This argument gives a theoretical foundation to the optimisation principle stated in the previous section.
	
	\ms{
	Moreover, we also show in Appendix~\ref{app:relEntropyEst2} that the determinant bounds the ability of retrieving any state close to the prior. In formulae,  this reads:
	\begin{align}
		D(\pi +\delta\rho ||\tilde\Phi\Phi(\pi +&\delta\rho ))\leq \nonumber\\
		\leq&\;\frac{D(\pi ||\pi+\delta\rho)}{2} \;\log\det(\tilde\Phi\Phi)^{-1},\label{eq:19}
	\end{align}
	where $\delta \rho$ is an arbitrary perturbation of the prior state such that $|\delta\rho|\ll 1$, and the inequality holds up to order $\mathcal{O}\norbra{|\delta\rho|^2}$.
	
	Finally, using similar arguments, we are also able to prove the following inequality (Appendix~\ref{app:relEntropyEst2}):
	\begin{align}
		\inf_{\rho,\sigma} \,&\norbra{\frac{D (\rho ||\sigma)  - D(\tilde\Phi\Phi(\rho)||\tilde\Phi\Phi(\sigma))}{D(\rho||	\sigma)}}\leq\nonumber\\
		&\qquad\qquad\qquad\qquad\qquad\leq2\log\det(\tilde\Phi\Phi)^{-1}.\label{eq:20}
	\end{align}
	In this way, the determinant can also be used to bound the maximum rate at which any two states become indistinguishable (the quantity in Eq.~\eqref{eq:20}). This is a well known quantifier of how much information is lost during the evolution $\tilde\Phi\Phi$~\cite{temmeH2divergenceMixingTimes2010}. }
	
	\section{Optimal State Retrieval}\label{sec:optimalStateRetrieval}
	We propose here an efficient algorithm to solve the maximisation of the determinant of $\tilde\Phi\Phi$ by reducing it to the problem of analytic centering. This can be expressed as follows: take a symmetric matrix $G[x]$ linearly dependent on some real scalars $\{x_i\}$ from a convex set $\mathcal{A}$.
	%\begin{align}
	%G[x] = G_0 + x_1G_1 +\dots + x_n G_n
	%\end{align}
	%where $x_i$ are real scalars from a convex set $\mathcal{A}$, and $G_i$ are symmetric matrices.
	The analytic centering problem corresponds to the minimisation:
	\begin{align}\label{eq:analyticalCentering}
		\min_{\substack{x\in\mathcal{A} \\ G[x]>0}}\;\; \log (\det (G[x])^{-1}).
	\end{align}
	This kind of problem can be efficiently solved on a computer~\cite{vandenbergheDeterminantMaximizationLinear1998, Grone84}. Moreover, assuming that the set of $x$ for which $G[x]>0$  is non-empty, and that the functional we are minimising in Eq.~\eqref{eq:analyticalCentering} is strictly convex, the solution is unique.
	
	We can now prove the reduction. First, it should be noticed that $\tilde\Phi\Phi$ is not symmetric in general, so the algorithm for the analytic centering cannot be directly applied. Define then the matrix:
	\begin{align}
		\Gamma[\lambda^{(\tilde \Phi)}]:&=\cJ_\pi^{-1/2}	\;(\tilde \Phi \Phi) \;\cJ_\pi^{1/2}= \\
		&=\sum_k\, \lambda_k^{(\tilde \Phi)} \cJ_\pi^{-1/2}\,V_{\pi|\Phi\pi}^{(k)}\,\cJ_{\Phi\pi}^{-1}\,\Phi\,\cJ_\pi^{1/2}.
	\end{align}
	\ms{It should be noticed that principle~(\ref{ax:4}) can be rewritten in matrix form as:
	\begin{align}
		(\tilde\Phi\,\Phi )\,\cJ_\pi= \cJ_\pi \, (\tilde\Phi\,\Phi)^T,\label{app:eq:detailed balance}
	\end{align}
	 from which it follows that }$\Gamma[\lambda^{(\tilde \Phi)}]$ is symmetric. \ms{Indeed}, the following holds:
	\begin{align}
		\Gamma[\lambda^{(\tilde \Phi)}]^T &= \cJ_\pi^{1/2}	\;(\tilde \Phi \Phi)^T \;\cJ_\pi^{-1/2} = \\
		&=\cJ_\pi^{1/2}\,\cJ_\pi^{-1}	\;(\tilde \Phi \Phi) \;\cJ_\pi\,\cJ_\pi^{-1/2} = \Gamma[\lambda^{(\tilde \Phi)}].
	\end{align}
	Moreover, thanks to the properties of the determinant we also have that:
	\begin{align}
		\det\Gamma[\lambda^{(\tilde \Phi)}] &= (\det \cJ_\pi^{-1/2})(\det \tilde \Phi \Phi)(\det \cJ_\pi^{1/2})=\nonumber\\
		&= \det \tilde \Phi \Phi.
	\end{align}
	In fact, since $\Gamma[\lambda^{(\tilde \Phi)}]$ and $\tilde \Phi \Phi$ are related by a similarity transformation, they actually share the same spectrum. \ms{This implies} that  the following optimisations are equivalent:
	\begin{align}
		&\max_{\substack{\tilde \Phi \,{\rm state\, retrieval}}} \;\; \det \tilde \Phi \Phi \iff\\
		&\quad\quad\max_{\substack{{\lambda}^{(\tilde \Phi)}_k\geq 0 ,\; \sum_k \!{\lambda}^{(\tilde \Phi)}_k = 1 \\ \Gamma[\lambda^{(\tilde \Phi)}]> 0}} \;\; \det \Gamma[\lambda^{(\tilde \Phi)}] \iff\\
		&\quad\quad\quad\quad\quad\min_{\substack{{\lambda}^{(\tilde \Phi)}_k\geq 0 ,\; \sum_k \!{\lambda}^{(\tilde \Phi)}_k = 1 \\ \Gamma[\lambda^{(\tilde \Phi)}]> 0}} \;\; \log(\det \Gamma[\lambda^{(\tilde \Phi)}]^{-1}) .
	\end{align}
	The last problem is the analytic centering for $\Gamma[\lambda^{\tilde{\Phi}}]$, which can be solved efficiently by means of convex optimisation. This concludes the reduction. 
	
	From the implementation of this algorithm, we obtained numerical evidence that the state retrieval so defined outperforms the Bayes inspired reverse not only on average, but at the single state level.  \js{To illustrate this, in Figure~\ref{fig:figcomposed} we plot the relative entropy of recovery $D\left(\rho\,\|\, \tilde{\Phi}\Phi(\rho)\right)$ for every state in the domain.
		% The relative entropy of recovery $D\left(\rho\| \tilde{\Phi}\Phi(\rho)\right)$ is defined as the relative entropy between a state $\rho$ and the state $\tilde{\Phi}\Phi(\rho)$ obtained by applying a chosen retrieval map $\tilde{\Phi}$ to the evolved state $\Phi(\rho)$.
	} 
		\ms{The results presented} corroborate the intuition that the retrieval map obtained by maximising the determinant of $\tilde\Phi\Phi$ is indeed better than the usual approach in the literature, i.e., Bayesian retrodiction.

	\section{Quantum retrieval map}
	
	The problem of identifying a state retrieval map for quantum dynamics is more subtle  than its classical counterpart. The Bayes' reversion, which depends on the existence of the joint probability of different observables in its derivation, has notoriously proven difficult to be extended to the quantum regime (see section \ref{sec:Bayes-rev} for a short review). For this reason, a reconstruction of a state retrieval map from physical principles is particularly suited to extend the concept of state recovery from the classical regime to quantum dynamics.
	
	Consider a completely positive and trace preserving (CPTP) map $\Phi$. The basic principles we require for a retrieval map to satisfy are the following:
	\begin{enumerate}
		\item \ms{The state retrieval is a CPTP map;}\label{ax:q1}
		\item \ms{If the map $\Phi$ is unitary, the state retrieval transformation is given by $\tilde\Phi := \Phi^{-1}$;}\label{ax:q2}
		\item The prior state should be perfectly retrieved, i.e., $\tilde\Phi (\Phi(\pi)) :=  \pi$.\label{ax:q3}
	\end{enumerate}
	These three principles already suffice to give a parametrisation of the recovery maps analogous to the one in Eq.~\eqref{eq:parametrisation}.

	%%%%
	\subsection{Parametrisation of CPTP maps with a given transition}	
	Given a CPTP map $\Psi$ with a fixed transition $\Psi(\pi)=\sigma$ we can decompose it as:
	\begin{align}\label{eq:QCrep}
		\Psi = \Lambda^{\Psi} \,\J^{-1}_{\pi},
	\end{align}
	where $\J_{\pi}$ is a completely positive generalisation of the multiplication by $\pi$, defined as $\J_{\pi}(\rho) := \sqrt{\pi} \rho \sqrt{\pi}$, and $\Lambda^\Psi$ is given by $\Lambda^\Psi = \Psi\, \J_{\pi}$. Since both $\Psi$ and $\J_{\pi}$ are CP, $\Lambda^\Psi$ is CP as well. Moreover, since $\Psi$ is trace preserving it follows that:
	\begin{align}
		(\Lambda^\Psi)^\dagger [\ids] = \J_{\pi}\,\Psi^\dagger[\ids] = \pi,\label{eq:TPConstraint}
	\end{align}
	because the trace preserving condition is equivalent to the equation $\Psi^\dagger[\ids] = \ids$. From the fixed transition it also follows that:
	\begin{align}
		(\Lambda^\Psi) [\ids] = \Psi\,\J_{\pi}\,[\ids] = \sigma.\label{eq:ax3Constraint}
	\end{align}
	In this way, similarly to what happens for the classical case, a quantum channel is uniquely identified by a map $\Lambda^\Psi\in\mathcal{U}_Q(\sigma,\pi)$, the space of CP transformations that map the identity to $\sigma$, and whose adjoint maps the identity to $\pi$. 
	This set is convex. Its extreme points can be characterised in terms of their Kraus operators $\{V_i\}_i$. In particular,  a map $\Psi[\rho] := \sum_i V_i \rho V_i^\dagger$ is an extreme point of $\mathcal{U}_Q(\sigma,\pi)$ if the following holds~\cite{Choi1975,rudolphExtremalQuantumStates2004}:
	\begin{enumerate}[label=\Alph*.]
		\item $\sum_{i}V_iV_i\da=\sigma$;
		\item $\sum_{i}V_i\da V_i=\pi$;
		\item $(V_iV_j\da)_{i,j}$ and $(V_j\da V_i)_{i,j}$ are jointly linear independent.
	\end{enumerate}
	
	Differently from the classical case, though, the set $\mathcal{U}_Q(\sigma,\pi)$ contains a non-trivial symmetry \ms{(that is, not reducible to a relabeling)}: consider the two unitary maps $U_\pi$ and $V_{\sigma}$, defined by $U_\pi[\rho]:= U \rho \,U^\dagger$, and satisfying  $U_\pi[\pi] = \pi$ (and analogously for $V_{\sigma}$, with $V_{\sigma}[\sigma] = \sigma$). Then Eq.~\eqref{eq:TPConstraint} and Eq.~\eqref{eq:ax3Constraint} are invariant under the transformation:
	\begin{align}
		\label{eq:unit-invariance}
		\Lambda^\Psi \rightarrow  V_\sigma\;\Lambda^\Psi \;U_{\pi}.
	\end{align}
	
	Hence, every $\Lambda^\Psi $ contained in $\mathcal{U}_Q(\sigma,\pi)$ is part of an invariant  family connected by the unitary transformations defined in Eq.~\eqref{eq:unit-invariance}. 
	
	The space $\mathcal{U}_Q(\sigma,\pi)$ can also be characterised in terms of the Choi operator $\mathcal{C}(\Lambda^{\Psi})$ of the maps $\Lambda^\Psi$ contained in it. In particular, we show in Appendix~\ref{app:Choi} how this  naturally translates to a characterisation of $\mathcal{U}_Q(\sigma,\pi)$ in terms of a marginal problem, leading to a set of linear inequalities that constrain the spectrum of the Choi matrices therein.

	%In the quantum regime, we can translate the fourth condition as:
	\ms{In order to extend principle~(\ref{ax:4}) to the quantum regime we generalise its matrix expression  (see Eq.~\eqref{app:eq:detailed balance}) as follows:}
	\begin{enumerate}[resume]
		\item \ms{The channel $\tilde\Phi \Phi$ satisfies the equation
			\begin{align}\label{eq:detailedBalanceQ}
				(\tilde\Phi\,\Phi )\,\J_\pi= \J_\pi \, (\tilde\Phi\,\Phi)^\dagger ,
			\end{align}
		with respect to the prior $\pi$.\label{ax:q4}}
	\end{enumerate}
	
	\ms{This expression is equivalent to a weak form of detailed balance for quantum evolutions~\cite{fagnolaGeneratorsKMSSymmetric2010, temmeH2divergenceMixingTimes2010}. In particular, it should be noticed that this principle coincides with the usual version of detailed balance for classical evolutions, as it can also be understood from the fact that for commuting states $\J_\pi\equiv \cJ_\pi$.}
	
	Finally, the last requirement can be translated to:
	\begin{enumerate}[resume]
		\item A map satisfying principles~(\ref{ax:q1}-\ref{ax:q4}) is a state retrieval if all the eigenvalues of $\tilde\Phi\,\Phi$ are non-negative.\label{ax:q5}
	\end{enumerate}

	It should be noticed that the spectrum of a CP-map is the same as the one of the corresponding vectorised version~\cite{wolf2008dividing}.
	%\ms{
	%	The invariant families of maps in $\mathcal{U}_Q(\sigma,\pi)$ correspond to all the maps sharing the same spectrum of their Choi matrix.
	%This allows us to identify each element of $\,\mathcal{U}_Q(\sigma,\pi)$ with the triplet $(\lambda^\Psi,U_{\pi},V_{\sigma})$ where $\lambda^\Psi$ is the ordered spectrum of $\mathcal{C}(\Lambda^{(\Psi)})$, and $U_{\pi},V_{\sigma}$ are the two unitary transformations  such that 
	%\begin{equation}
	%	\mathcal{C}(\Lambda^{(\Psi)})=(U_{\pi}\otimes V_{\sigma})\,\cJ_{\lambda^\Psi}\,(U\da_{\pi}\otimes V\da_{\sigma}).
	%\end{equation}

	%Given all these considerations, this parametrisation gives a way to uniquely identify a map $\Psi$. 
	%In particular we have that 
	%\begin{equation}
	%	\Psi[\rho] = \tr_A\sqrbra{((\J^{-1}_{\Psi\pi}\rho)^T\otimes\id_B)\,(U_{\pi}\otimes V_{\Psi\pi})\cJ_{(\lambda^\Psi)}(U\da_{\pi}\otimes V\da_{\Psi\pi})}.
	%\end{equation}}

	%%%%	
%	\subsection{Parametrisation of quantum state retrieval maps}\label{sec:qParametrisation}
%	\ms{La mettiamo?.}
%	The parametrisation of the family quantum channels with fixed transition is more complex than the one obtained in the classical case.  In fact, for two general quantum states $\pi$ and $\sigma$ we just have a characterisation of the extreme points of the convex set $\mathcal{U}_Q(\sigma,\pi)$ \cite{rudolphExtremalQuantumStates2004}.
%	\ms{Do we have vertices? Equivalence classes with unitaries? Magari smooth ma tutti elementi combinazione di solo n punti estremi, ma con tante possibili combinazioni diverse perche' ho unitarie. }
	
	\subsection{Petz' map}
	We can now proceed to define a map analogous to the Bayes inspired reverse for quantum systems. First, it is clear from Eq.~\eqref{eq:TPConstraint} and Eq.~\eqref{eq:ax3Constraint} that \ms{for any generic} $\Lambda^\Psi$ \ms{in} $ \mathcal{U}_Q(\pi,\Phi\pi)$, then $(\Lambda^\Psi)^\dagger\in\mathcal{U}_Q(\Phi\pi,\pi)$, so there is a one to one correspondence between the two sets, given by the adjoint transformation. Moreover, the CPTP map $\Phi$ can be written as:
	\begin{align}
		\Phi = \Lambda^\Phi\, \J_{\pi}^{-1},
	\end{align}
	where $\Lambda^\Phi\in\mathcal{U}_Q(\Phi\pi,\pi)$.  By inspecting Eq.~\eqref{eq:baDecomposition}, one can see that for classical systems the Bayes' retrodiction is obtained by choosing $\Lambda^{\tilde\Phi} := (\Lambda^{\Phi})^T$. In complete analogy we define:
	\begin{align}
		\tilde\Phi_P = (\Lambda^\Phi)^\dagger \,\J_{\Phi\pi}^{-1} = \J_{\pi}\, \Phi^\dagger\,\J_{\Phi\pi}^{-1},\label{eq:petzRecovery}
	\end{align}
	where on the right hand side one can read the definition of the Petz recovery map, commonly used as a quantum extension of the Bayes rule~\cite{watanabe1955,watanabe1965,Leifer2013,buscemi2021,aw2021}. This argument gives yet another derivation justifying this identification.

	It is easy to show that the Petz recovery map satisfies all the desiderata of a state retrieval map. In particular one notices that the Petz recovery map satisfies principles (\ref{ax:q4}) and (\ref{ax:q5}) by rewriting it as:
	\begin{align}
		%\tilde{\Phi}_{PR} \,\Phi= U_\pi\, \J_{\pi} \, \sqrbra{(\J_{\Phi\pi}^{-1/2}\,V_{\Phi\pi}^{1/2}\,\Phi)^\dagger(\J_{\Phi\pi}^{-1/2}\,V_{\Phi\pi}^{1/2}\,\Phi)},
		\tilde{\Phi}_{P} \,\Phi=  \J_{\pi} \, \sqrbra{(\J_{\Phi\pi}^{-1/2}\,\Phi)^\dagger(\J_{\Phi\pi}^{-1/2}\,\Phi)},
	\end{align}
	and by using similar arguments as the one for the classical case.
	%given that $U_\pi$ and $V_{\Phi\pi}$ have positive spectrum.
	%\ms{Così a naso questa cosa è possibile solo per l'identità. A quel punto manco serve nominarla la versione rotata (oppure la si nomina per dire che non funziona)}
	
	%In this way, we can define a class of reverse maps which contains the most used extension of the Bayes' rule to the quantum regime. This promising result will be further explored in future works.
	This discussion shows that not only can the approach presented here  be useful to clarify the basic requirements for a quantum state retrieval map, but it can also help in highlighting the correspondence between the classical and the quantum scenario.

	\subsection{Optimal state retrieval: case studies}
	In complete analogy with the classical case we define the optimal retrieval map to be the one satisfying the following
	\begin{principle*}
		The optimal retrieval map is defined to be the $\tilde{\Phi}_O$ that maximises the determinant of $\tilde{\Phi}_O\,\Phi$ under the constraints (\ref{ax:q1}-\ref{ax:q5}).
	\end{principle*}
	\begin{figure}
		\centering
		\includegraphics[width=0.9\linewidth]{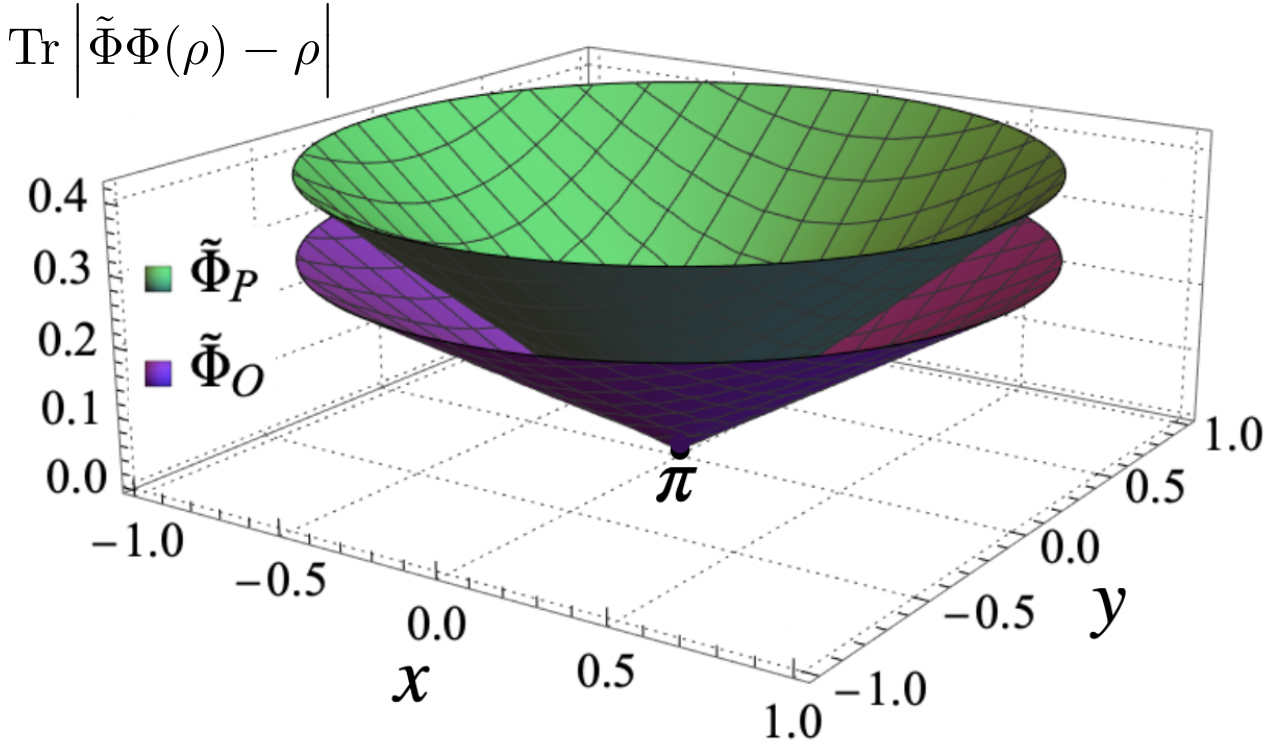}
		\caption{Trace distance between a distribution and its evolution forwards and backwards for $\Phi = \Delta_\eta$ for a qubit, using as prior distribution $\pi = \ids/2$. The states are parametrised as $\rho = (\ids + x \sigma_x + y\sigma_y)/2$, corresponding to the disk at the equator of the Bloch sphere. It can be seen how the optimal state retrieval map outperforms the Petz recovery on all states. It should be pointed out that the relative entropy presents the same feature, but the trace distance makes the plot more understandable.}
		\label{fig:fig2}
	\end{figure}

	The use of the volume as a significant quantity in the study of quantum channels has been already explored in relevant works such as \cite{paternostro2013,Buscemi2019}. It is not immediately clear how one could devise a parametrisation to explore the whole space $\mathcal{U}_Q(\sigma,\pi)$. Moreover, the symmetry expressed by Eq.~\eqref{eq:unit-invariance} makes designing a maximisation algorithm more involved. For this reason, we limit ourselves here to the treatment of analytically solvable cases.
	
	In particular, consider the depolarising channel given by:
	\begin{align}
		\Delta_\eta(\rho) = (1-\eta) \,\rho + \eta \,\frac{\ids}{d},
	\end{align}
	where $\eta$ is a scalar parameter in $[0,1+(d^2-1)^{-1}]$ and $d$ is the dimension of the quantum system in consideration. Choosing $\ids/d$ to be the prior state, all the calculations can be carried out analytically.
	
	First, we compute the Petz recovery map in this case. The prior state is invariant under the transformation, that is $\Delta_\eta(\ids/d) \equiv \ids/d$, implying that $\J_{\ids/d} = \J_{\Delta_\eta\ids/d}$. This is simply given by $\J_{\ids/d} = \id/d$, where we used a different notation for the identity superoperator $\id$ and the state $\ids/d$. Finally, we can compute the adjoint of the depolarising channel from the series of equations:
	\begin{align}
		\Tr{\sigma^\dagger \Delta_\eta (\rho)} &= (1-\eta) \,\Tr{\sigma^\dagger\rho} + \frac{\eta}{d} = \\
		&=\Tr{\left((1-\eta)\,\sigma+\eta \,\frac{\ids}{d}\right)^\dagger\rho} = \\
		&= \Tr{\Delta_\eta (\sigma)^\dagger \rho},
	\end{align}
	implying that $\Delta_\eta^\dagger = \Delta_\eta$. Hence, by using the definition in Eq.~\eqref{eq:petzRecovery} we obtain that the Petz recovery map for the depolarising channel is given by:
	\begin{align}
		(\tilde\Delta_\eta)_P = \J_{\ids/d}\, \Delta_\eta^\dagger \,\J_{\Delta_\eta\ids/d}^{-1} =\Delta_\eta, \label{eq:petzRecDepolarising}
	\end{align}
	that is by the depolarising channel itself. 
	This was somehow expected, for, as already observed in \cite{aw2021}, the Bayes inspired reverse channel computed considering as prior a fixed point of the channel is the channel itself.
	\begin{figure}
		\centering
		\includegraphics[width=0.9\linewidth]{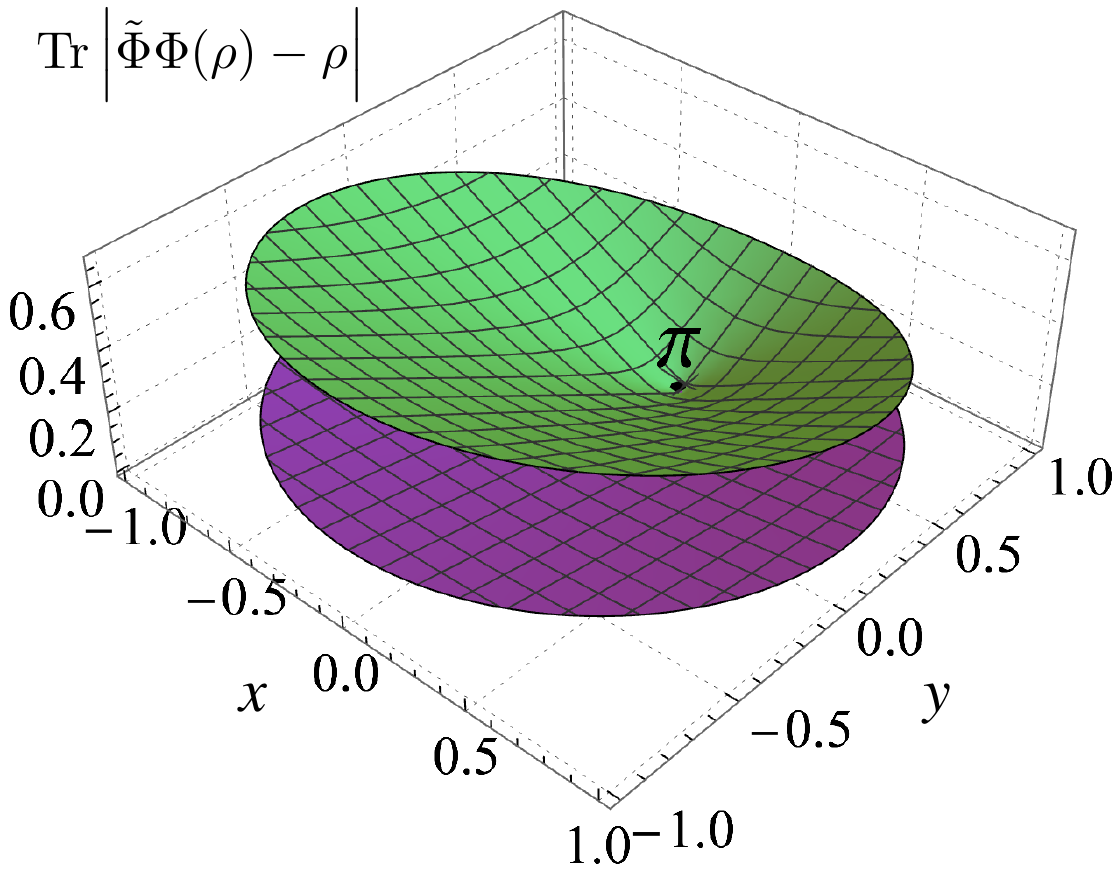}
		\caption{\ms{Trace distance between a distribution and its evolution forwards and backwards for the evolution $\Phi$ specified in Eq.~\eqref{eq:phiThermalSwap} and states of the form $\rho = ((\ids + x \sigma_x + y\sigma_y)/2 )\otimes \gamma_\beta$ and prior state $\pi = \gamma_\beta \otimes\gamma_\beta$.}}
		\label{fig:Termico}
	\end{figure}	
	We can now pass to compute the optimal state retrieval. There are two remarks that need to be made beforehand: first, it should be noticed that the constraint in Eq.~\eqref{eq:detailedBalanceQ} is satisfied at the level of the map itself, that is
	\begin{align}
		\Delta_\eta \,\J_{\ids/d}= \J_{\ids/d} \,\Delta_\eta^\dagger.
	\end{align}
	Moreover, the spectrum of $\Delta_\eta$ is real and positive, as it can be understood by decomposing it on any basis of the Hermitian operators. These two observations together imply that 
	\begin{align}
		(\tilde\Delta_\eta)_O = \id.
	\end{align}
	In fact, this map always maximises the determinant of $\tilde{\Phi}_O\,\Phi$, since any other CPTP will contract the volume of the phase space. Usually, though, it is ruled out by the requirements imposed by principles ~(\ref{ax:q4}) and~(\ref{ax:q5}). The generality of these considerations directly leads to the following:
	\begin{theorem*} Whenever a transformation \ms{$\Phi$} has positive spectrum and it is detailed balance with respect to the prior state (meaning that $\Phi\, \J_\pi = \J_\pi\, \Phi^\dagger$) the optimal state retrieval is given by the identity map.
	\end{theorem*}
	\begin{figure*}[t!]
		\centering
		\includegraphics[width=1.\linewidth]{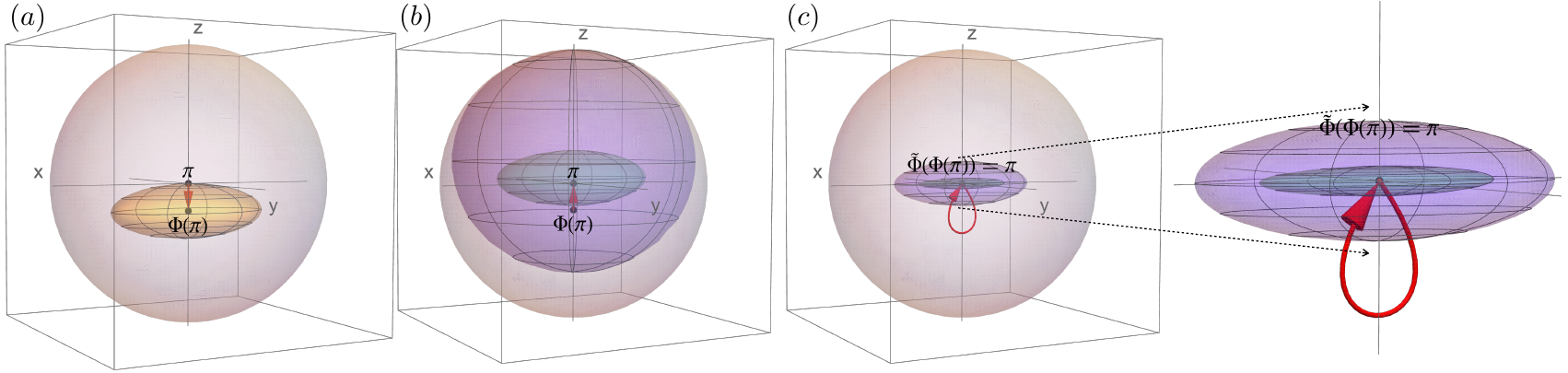}
		\caption{\js{Comparison of the action on the Bloch Sphere of the Petz map and of the optimal retrieval map. In all panels the shadowed area without wires represents the Bloch sphere. In panel $(a)$ we plot in yellow the image of the Bloch sphere under the action of the map $\Phi$. With the red arrow we highlight the specific transition $\pi=\frac{\ids}{2}\to\Phi(\pi)$. In panel $(b)$ we plot the action on the Bloch sphere of the optimal ($\tilde{\Phi}_O$) and Petz ($\tilde{\Phi}_P$) retrieval maps in purple and green, respectively. In both cases the chosen prior is $\pi=\frac{\ids}{2}$. With the red arrow we highlight how both $\tilde{\Phi}_O$ and $\tilde{\Phi}_B$ map the state $\Phi(\pi)$ to $\tilde{\Phi}_O(\Phi(\pi))=\tilde{\Phi}_B(\Phi(\pi))=\pi$. The Bloch sphere is compressed to a much smaller image by the action of the Petz map compared with the optimal retrieval map. Panel $(b)$ helps visualise part of the rationale for the criteria characterising the optimal retrieval map. In this case, where the map $\Phi$ is simply a translation composed with a compression, the optimal retrieval map is also the composition of a compression and a translation specified as follows: the translation is the one recovering the desired prior $\pi$ (the red arrow in the picture) and the compression is the minimal one making $\tilde{\Phi}_O$ physical (i.e., so that the image of $\tilde{\Phi}_O$ is contained in the Bloch sphere). In panel $(c)$ we plot the action on the Bloch sphere of the forth-and-back maps  $\tilde{\Phi}_O \Phi$ and $\tilde{\Phi}_P \Phi$ (with the same colour scheme as before). The prior $\pi$ is the fixed point of the forth-and-back maps, and the different magnitude in the compression of the Bloch sphere through $\tilde{\Phi}_O \Phi$ and $\tilde{\Phi}_P \Phi$ is evident. The choice of the prior $\pi=\frac{\ids}{2}$ makes $\tilde{\Phi}\Phi$ unital. This allows us to use the parametrisation given in~\cite{ruskai2002} to explore the whole space of possible state retrieval maps.}}
		\label{fig:Traslazione-Composed}
	\end{figure*}
	
	The theorem, means that in this case the optimal strategy is to leave the system unperturbed. It should be noticed that under the same assumptions the Petz recovery map is given by the map itself, $\tilde\Phi_P = \Phi$, so that applying it leads to a further deterioration of the information on the initial state. This shows how our definition of optimal retrieval is more suited in the task of recovering a state after a transformation. The difference in performance between the Petz recovery map and the optimal state retrieval is shown in Fig.~\ref{fig:fig2}.
	
	A crucial simplification in the study of the depolarising channel is that it is a unital channel, so that one can use $\ids/d$ as a prior state, leading to  $\J_{\ids/d} = \id/d$. In the following we show that one can obtain analytical insights even in the case of non-unital maps.
	In Fig.~\ref{fig:Termico} we compare the performance of the optimal map and the Petz' one for the two-qubit channel defined by:
	\begin{align}
		\Phi(\rho_A\otimes\rho_B) = \theta_{\lambda_1}(\rho_B)\otimes\theta_{\lambda_2}(\rho_A),
	\end{align}
	where $\theta_\lambda$ is the thermalising channel defined by:
	\begin{align}\label{eq:phiThermalSwap}
		\theta_\lambda (\rho) = (1-\lambda)\rho + \lambda \Tr{\rho} \gamma_\beta
	\end{align}
	and $\gamma_\beta=\frac{e^{-\beta H}}{\Tr{e^{-\beta H}}}$ is the Gibbs state associated to the Hamiltonian $H :=\epsilon\ketbra{1}{1}$. Then, a simple calculation shows that the Petz map coincides with the original channel, i.e., $\tilde \Phi_P = \Phi$. On the other hand, the map maximising the determinant (under the constraints ax.~(\ref{ax:q1}-\ref{ax:q5}))  is given by $\tilde\Phi_O = \text{SWAP}$, the swap operator. 
	
	\js{Finally, in Figure~\ref{fig:Traslazione-Composed} we highlight part of the rationale for the criteria characterising the optimal retrieval map. In order to do so we study the emblematic example of a map $\Phi$ obtained from the composition of a translation and a compression in the Bloch sphere. From panel $(b)$ it is evident how the optimal retrieval map corresponds to the map that minimises the compression of the domain while recovering the desired prior $\pi$.
		%\textcolor{red}{\sout{The criteria characterising the optimal retrieval map are designed for all maps $\Phi$. Thus in the case of the Bloch sphere, even more complex map, as for example maps including rotations or flips are taken in consideration.} LA METTO?} 
	}

	\subsection{Quality of the retrieval}
	
	As we did for stochastic maps, we present here some analytical arguments suggesting that optimising the determinant indeed leads to a better quality of retrieval. 
	
	First, it should be noticed that for quantum channels Eq.~\eqref{eq:relEntropyIdFB} applies without modifications, so the same arguments presented above in this regard can also be applied to quantum dynamics.
	
	The generalisation of Eq.~(\ref{eq:19}-\ref{eq:20}), instead, needs a bit more care. First, we introduce the following contrast function:
	\begin{align}
			H_{{\rm sq}}(\rho||\sigma) =  \Tr{\sqrt{\rho}(\rho-\sigma)\sqrt{\sigma^{-1}}}.
	\end{align}
	This quantity is positive, zero if and only if $\rho \equiv\sigma$, and can be regarded as akin to the Kullback--Leibler relative entropy. It is connected to the super-operator $\J_\rho$ thanks to the following expansion for close-by states:
	\begin{align}
		H_{{\rm sq}}(\rho||\rho+\delta\rho) \simeq \frac{1}{2} \Tr{\delta \rho\,\J_{\rho}^{-1}[\delta \rho]},
	\end{align}
	for $\Tr{|\delta \rho|}\ll1$. Then, Eq.~\eqref{eq:19} get replaced by:
	\begin{align}
		H_{{\rm sq}}(\pi +\delta\rho ||&\tilde\Phi\Phi(\pi +\delta\rho ))\leq\nonumber \\
		\leq&\;\frac{H_{{\rm sq}}(\pi ||\pi+\delta\rho)}{2} \;\log\det(\tilde\Phi\Phi)^{-1},
	\end{align}
	(up to order $\mathcal{O}\norbra{|\delta\rho|^2}$) and Eq.~\eqref{eq:20} get replaced by:
	\begin{align}
		\inf_{\rho,\sigma} \,&\norbra{\frac{H_{{\rm sq}}(\rho ||\sigma)  - H_{{\rm sq}}(\tilde\Phi\Phi(\rho)||\tilde\Phi\Phi(\sigma))}{H_{{\rm sq}}(\rho||	\sigma)}}\leq\nonumber\\
			&\qquad\qquad\qquad\qquad\qquad\leq2\log\det(\tilde\Phi\Phi)^{-1}.
	\end{align}
	The proof for these inequalities is completely analogous to the one for the classical case and it is presented in Appendix~\ref{app:relEntropyEst2}. The main difference with the classical case is that here we cannot consider the Umegaki relative entropy, unless we demand a stronger version of principle~(\ref{ax:q4}), but this seems unnecessary for the situation at hand (see Appendix~\ref{app:relEntropyEst2} for more details).

	\section{Bayes reversion from physical principles}	\label{sec:Bayes-rev}
	
	The classic derivation of the Bayes inspired reverse channel comes directly from fundamental theorems of probability theory. \js{In fact, since the intersection of two sets $A$ and $B$ is commutative, this means that $P(A\cap B)=P(B \cap A)$, so by using the rule of conditional probability (or the axiom of conditional probability following de Finetti \cite{finetti1974}) one easily obtains Bayes' theorem}.
	This derivation heavily relies on the notion of commutativity for the operation of composing probabilities, which is unavailable when trying to extend the construction of a reverse channel from classical to quantum probabilities. In fact, the non-commutative structure at the basis of quantum theory makes the assignment of a compound probability for a general pair of quantum events problematic. In order to obtain a quantum extension of Bayes inspired reversion a different approach is needed, and many attempts already exist.
	Among the most modern ones we mention two. The first obtains the classical Bayes inspired reverse from entropy maximisation methods; an overview about this topic is given in~\cite{Giffin07}. This approach has been further developed to the quantum case as in~\cite{Ali12,kostecki2014,vanslette2017,vanslette2018}. Here the Bayes inspired reverse is mainly treated as a tool from inference problems and its physical relevance is somehow set aside.
	
	The second modern and promising approach starts from giving a definition of Bayes inspired reverse in the language of category theory. For its generality this approach is naturally extensible to the quantum scenario, as it is shown in~\cite{parzygnat2020} where they give a  characterisation of Bayes inspired reverse in terms of commuting diagrams and they show its meaning both in classical and quantum probability. Similar approaches can be found in~\cite{Accardi78,Accardi82,Leifer06,coecke12,culbertson14}.
	
	In this section, motivated by the results presented so far, we are interested in exploring the possibility of a reconstruction of the Bayes inspired reverse starting from few physical principles. If this would be doable, the extension from classical to quantum probability would result \ms{naturally}, as it was shown in the previous section.
	
	The $5$ requirements presented thus far only individuate a family of state retrieval maps which includes the Bayes inspired reverse as a particular case. We can then try to add an additional requirement to see if this singles out the Bayes inspired reverse map \ms{in the classical case}. A particularly natural choice is the following:
	\begin{enumerate} [resume]
		\item The reversion procedure is involutive, that is $\tilde{\tilde{\Phi}}=\Phi$.\label{ax:6}
	\end{enumerate}
	
	As we argued in the introduction, we call the state retrieval maps that satisfy this principle reverse maps.\\
	Principle \eqref{ax:6} implies that $\R^2 = \id$, which heavily constrains the freedom on the choice of the reversion procedure $\R$. In the next section we present some evidence that allow us to conjecture that the requirement~(\ref{ax:6}) is strong enough to single out the identity transformation (corresponding to the Bayes inspired reverse) at least in the case in which $\R$ is linear and solely depends on the unordered pair of states of the fixed transition. %This last strong limitation we impose is mainly justified by the necessity of making the problem more tractable. .\\
	
	\begin{figure*}[t!]
		\centering
		\includegraphics[width=1.\linewidth]{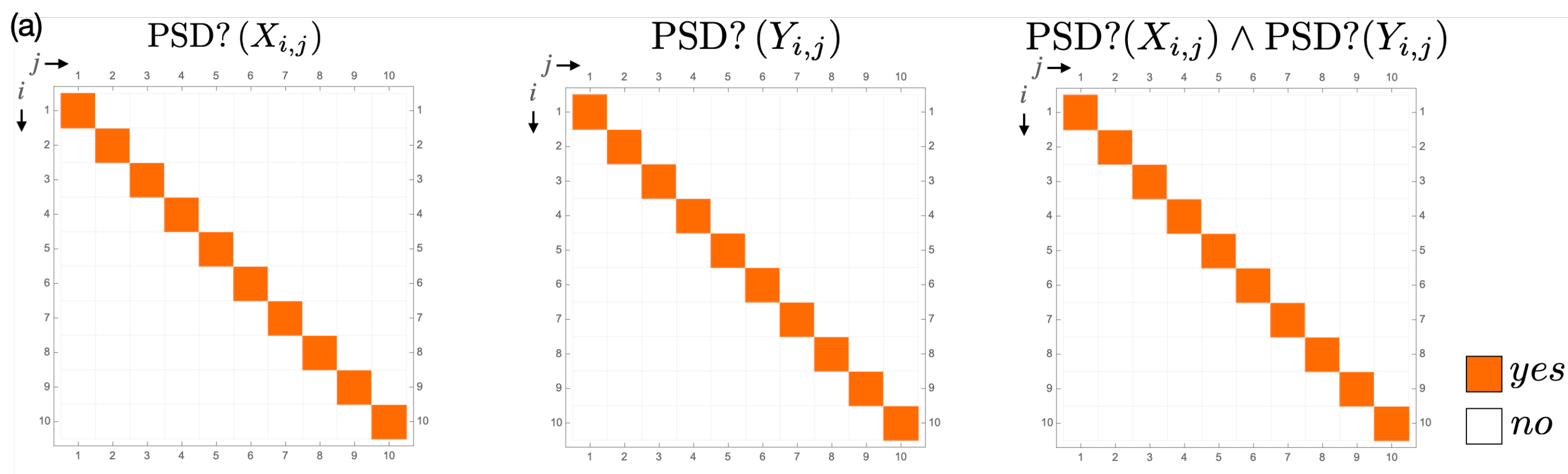}
		\includegraphics[width=1.\linewidth]{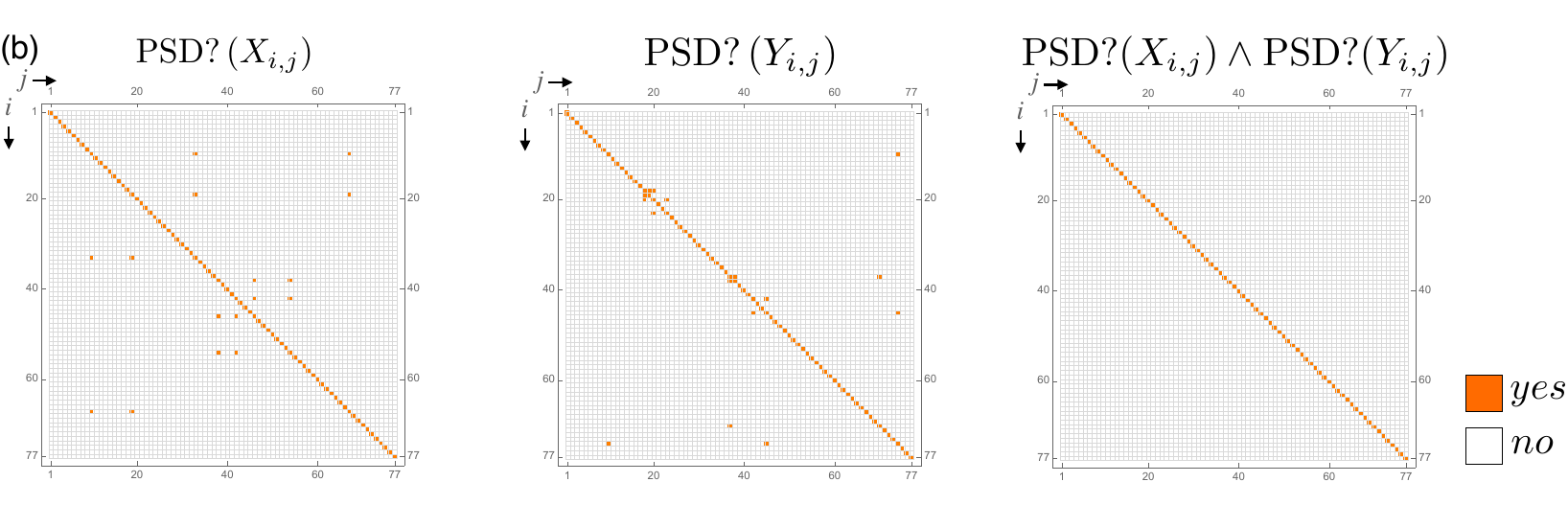}
		\caption{\js{Graphical representation of the tests \eqref{eq:set-X} and \eqref{eq:set-Y}.  We check if each matrix $\{X_{i,j}\}_{i,j}$ and $\{Y_{i,j}\}_{i,j}$ is PSD. Each of the six tables represents the answer to the question in their title. An orange square represents a "Yes", while a white square represents a "no".  For example, the first orange square in the first table in the top left tells us that the matrix $X_{1,1}$ is PSD.
		In order to check if $\mathcal{R}$ can be any other permutation other than the identity one checks which of the matrices $X_{i,j}$ and $Y_{i,j}$ are Positive SemiDefinite (PSD). If, for a fixed value of $i$ and $j$, both the matrix $X_{i,j}$ and $Y_{i,j}$ are PSD, then a $\mathcal{R}$ permuting $i$ and $j$ is allowed. In figure (a) we choose $\pi=({0.1, 0.2, 0.7})$ and $\Phi\pi=({0.3, 0.6, 0.1})$ (notice that $\Phi\pi$ is jut a notation for a vector, and does not refer to any map in particular) and compute the vertices $\{V^{(i)}_{\Phi\pi|\pi}\}_i$ of $\mathcal{U}(\Phi\pi,\pi)$ using the algorithm of Jurkat and Ryser~\cite{Jurkat1967}. From the vertices we can compute the matrices $X_{i,j}$ and $Y_{i,j}$ and check if they are PSD. In the left and central plots of figure (a) we use an orange square to denote a PSD matrix and a white square to denote a matrix that is not PSD. We note that only the matrices $\{X_{i,i}\}_i$ and $\{Y_{i,i}\}_i$ are PSD, thus, in this case, $\mathcal{R}$ is not allowed to be any permutation different from the identity. In the rightmost plot of figure (a)  the color of the square at coordinate $(i,j)$ is orange if both $\{X_{i,i}\}_i$ and $\{Y_{i,i}\}_i$ are PSD, white otherwise. In figure (b) we chose $\pi=(0.1, 0.6, 0.1, 0.2)$ and $\Phi \pi = (0.1, 0.2, 0.3, 0.4)$ and plotted the analogous quantities of figure (a). In this case, checking the positive semidefinitness of just $\{X_{i,i}\}_i$ or $\{Y_{i,i}\}_i$  is not enough to isolate just Bayes reversion. Checking the simultaneous positive semidefinitness of $\{X_{i,i}\}_i$ and $\{Y_{i,i}\}_i$ isolates, even in this case, just the Bayes reversion.  }}
		\label{fig:PSDQ}
	\end{figure*}
	
	\subsection{Characterisation of $\R$}
	As it was shown in Section~\ref{sec:bayesreverse}, Bayes' reversion corresponds to choosing the transformation $\R$ to be the identity on the space of coefficients. We are thus interested in knowing if principles~(\ref{ax:1}-\ref{ax:6}) are enough to ensure that $\R = \id$, at least in the case in which $\R$ is a matrix. 
	
	We order the vertices of $\mathcal{U}(\Phi\pi,\pi)$ in the following way: any vertex that corresponds to a permutation is moved to the beginning of the list $\{V^{(i)}_{\Phi\pi|\pi}\}_{i=1,\dots,n}$. Say there are $\ell$ of those. Then, we have the following 
	
	\begin{observation*}
		$\R$ is the direct sum of the identity matrix acting on the first $\ell$ sites and a permutation matrix with cycles of maximal length $2$ acting on sites $\ell+1,\dots,n$.
	\end{observation*}
	
	\begin{proof}
		Thanks to the structure of $\mathcal{U}(\Phi\pi,\pi)$, one can interpret the coefficients $\{\lambda_k^{({\Phi})}\}$ as a probability vector. Thus $\R$ must map probability distributions into probability distributions, meaning that $\R$ is a stochastic matrix. Moreover, principle~(\ref{ax:6}) implies $\R^2 =\id$, meaning that $\R$ is invertible and coincides with its inverse. It should be noticed that all the invertible stochastic matrices are permutations. The involutive principle then also implies that it must be a permutation of cycle at most $2$. We can now focus on the action of $\R$ on the first $\ell$ indices. From principle~(\ref{ax:2}) we know that permutations must be mapped into their inverse, that is $U\rightarrow U^T$. Thanks to the relation between the vertices of $\mathcal{U}(\Phi\pi,\pi)$ and $\mathcal{U}(\pi,\Phi\pi)$ this corresponds to $\R$ acting as the identity on the first $\ell$ elements of  $\{\lambda_k^{({\Phi})}\}$.
	\end{proof}
	
	Since $\R$ is a stochastic matrix, it is sufficient to study its action on the vertices of the simplex of the probability vectors  $\{\lambda_k^{({\Phi})}\}$. In particular we need to check if there is any permutation of two vertices of this simplex that is admissible other \ms{than} the identity. 
	
	We focus on the action of $\R$ on single vertices. Consider in particular the case in which $\Phi :=V^{(i)}_{\Phi\pi|\pi}\mathcal{J}^{-1}_{\pi}$. Since $\R$ is a permutation, there exists a vertex  $V^{(j)}_{\pi|\Phi\pi}$ which satisfies $\tilde\Phi \equiv V^{(j)}_{\pi|\Phi\pi} \mathcal{J}^{-1}_{\Phi \pi}= (V^{(j)}_{\Phi\pi|\pi})^T \mathcal{J}^{-1}_{\Phi \pi}$. From principles~(\ref{ax:4}) and~(\ref{ax:5}) the following matrix
	\begin{equation}
		\label{eq:set-X}
		X_{i,j}=\cJ^{-\frac{1}{2}}_\pi (V^{(j)}_{\Phi\pi|\pi})^T \cJ^{-1}_{\Phi(\pi)} V^{(i)}_{\Phi\pi|\pi} \cJ^{-\frac{1}{2}}_\pi
	\end{equation}
	is positive semidefinite.
	
	At the same time, due to principle~(\ref{ax:6}), if the vertex $(V^{(j)}_{\Phi\pi|\pi})^T$ corresponds to the inverse of $V^{(i)}_{\Phi\pi|\pi}$, then $V^{(i)}_{\Phi\pi|\pi}$ must be the inverse of $(V^{(j)}_{\Phi\pi|\pi})^T$. This consideration, together with principle~(\ref{ax:4}) and~(\ref{ax:5}), then also implies that the matrix
	\begin{equation}
		\label{eq:set-Y}
		Y_{i,j}=\cJ^{-\frac{1}{2}}_{\Phi(\pi)} V^{(i)}_{\Phi\pi|\pi} \cJ^{-1}_{\pi} (V^{(j)}_{\Phi\pi|\pi})^T \cJ^{-\frac{1}{2}}_{\Phi(\pi)}
	\end{equation}
	is positive semidefinite.
	
	Since the number of vertices is finite, it is easy to explicitly verify for which set of indices Eq.~\eqref{eq:set-X} and Eq.~\eqref{eq:set-Y} are positive semidefinite. We verified \ms{this} for many possible families of stochastic maps and found that the only admissible $\R$ is the identity, meaning that principle~(\ref{ax:6}) seems to be enough to single out the Bayes reversion \js{(see Figure \ref{fig:PSDQ})}. Despite this promising result, an analytical proof of this fact is still missing. \js{In fact we miss a characterisation of the properties of the vertices for generic $\mathcal{U}(\Phi\pi,\pi)$. To the best of our knowledge, for an arbitrary pair $(\Phi\pi,\pi)$, it is not even possible to know the precise number of vertices of the set $\mathcal{U}(\Phi\pi,\pi)$ without first mechanically constructing them using the algorithm of Jurkat and Ryser~\cite{Jurkat1967}}.

	\section{Conclusions}
	In the present work, we addressed the problem of finding an optimal strategy for the retrieval of a state after the evolution induced by a physical map. We assumed to have a full characterisation of the physical map on the system (given for classical systems in terms of a left stochastic matrix, and for quantum systems in terms of a CPTP map) and we \ms{wanted} to find a physical transformation ascribable to some reverse transformation. 
	
	To this end, we postulated five physically motivated principles that all retrieval maps should satisfy:  \eqref{ax:1}~they are physical; \eqref{ax:2}~on invertible maps they give the inverse;  \eqref{ax:3}~ they perfectly retrieve a fiducial state $\pi$; \eqref{ax:4}~the transformation $\tilde \Phi \Phi$ mapping forward and backwards is detailed balanced with respect to $\pi$; and \eqref{ax:5}~the eigenvalues of $\tilde \Phi \Phi$ are positive. We showed that both the Bayes inspired reverse, in the classical case, and the Petz recovery map, in the quantum one, satisfy all these principles. 
	
	After giving a parametrisation of the maps compatible with the requirements above, we defined a retrieval to be a transformation $\R$ associating to the pair $\Phi$ and $\pi$ a state retrieval map $\tilde\Phi$. In this context, the map $\R$ corresponding to the Bayes inspired reverse and the Petz recovery takes a particularly simple form: namely, it corresponds to the identity on the coefficients parametrising the possible retrieval maps. 
	
	At this point, we proposed a maximisation principle to define the optimal state retrieval. This seems to outperform the Bayes inspired reverse, or the Petz recovery, both on average and at the level of the single state. We complement the numerical evidence supporting this fact with analytical intuitions about why this is the case.
	
	Finally, in the last section of the paper, we investigated the possibility of singling out the Bayes inspired reverse among the possible state retrievals by adding an additional principle. We propose as a candidate the following: \eqref{ax:6}~the retrieval of the retrieval is the original map. This principle is motivated by interpreting state retrieval as a generalisation of the time inversion. Despite not being able to prove that this is enough to isolate the Bayes inspired reverse, we have strong numerical suggestions supporting the claim.
	
	Apart from settling down the question whether principle~\eqref{ax:6} is enough to isolate Bayes' reversion, there are a number of subtleties in the quantum regime that we did not explore. Primarily, there is some arbitrariness in the choice of $\J_\pi$: our choice was motivated by the fact that both $\J_\pi$ and its inverse are CP~\cite{petzIntroductionQuantumFisher2011}. Unfortunately, different choices of $\J_\pi$ impose inequivalent characterisations of the detailed balance in principle~\eqref{ax:4}~~\cite{fagnolaGeneratorsKMSSymmetric2010, temmeH2divergenceMixingTimes2010}. For this reason, it will be interesting to study what role this choice has in the definition of reverse maps~\cite{tim22b}. Moreover, since the concepts of retrodiction and reverse processes increasingly seem to play a fundamental role in thermodynamics \cite{aw2021,buscemi2021,BuscemiRotondo2020,tim22a}, it would be interesting studying the role of the complete family of state retrievals.  Finally, it is not directly clear how one could extend the algorithm for the classical scenario to the quantum case. These questions need a treatment of their own and are therefore left for future research.
	
	\emph{Acknowledgements.} We are grateful to F. Buscemi for the useful comments to the first draft of the manuscript. This project has received funding from the European Union’s Horizon 2020 research and innovation programme under the Marie Skłodowska-Curie grant agreement No 713729,
	and from the Government of Spain (FIS2020-TRANQI and Severo Ochoa CEX2019-000910-S), Fundacio Cellex, Fundació Mir-Puig, Generalitat de Catalunya (SGR 1381 and CERCA Programme).

	\bibliographystyle{unsrtnat}
	\bibliography{Biblio.bib}

\begin{thebibliography}{50}
\providecommand{\natexlab}[1]{#1}
\providecommand{\url}[1]{\texttt{#1}}
\expandafter\ifx\csname urlstyle\endcsname\relax
  \providecommand{\doi}[1]{doi: #1}\else
  \providecommand{\doi}{doi: \begingroup \urlstyle{rm}\Url}\fi

\bibitem[Watanabe(1955)]{watanabe1955}
Satosi Watanabe.
\newblock Symmetry of {{Physical Laws}}. {{Part III}}. {{Prediction}} and
  {{Retrodiction}}.
\newblock \emph{Rev. Mod. Phys.}, 27\penalty0 (2):\penalty0 179--186, April
  1955.
\newblock \doi{https://doi.org/10.1103/RevModPhys.27.179}.

\bibitem[Watanabe(1965)]{watanabe1965}
Satosi Watanabe.
\newblock Conditional {{Probability}} in {{Physics}}.
\newblock \emph{Progress of Theoretical Physics Supplement}, E65:\penalty0
  135--160, January 1965.
\newblock \doi{https://doi.org/10.1143/PTPS.E65.135}.

\bibitem[Buscemi and Scarani(2021)]{buscemi2021}
Francesco Buscemi and Valerio Scarani.
\newblock Fluctuation theorems from {{Bayesian}} retrodiction.
\newblock \emph{Phys. Rev. E}, 103\penalty0 (5):\penalty0 052111, May 2021.
\newblock \doi{https://doi.org/10.1103/PhysRevE.103.052111}.

\bibitem[Aw et~al.(2021)Aw, Buscemi, and Scarani]{aw2021}
Clive~Cenxin Aw, Francesco Buscemi, and Valerio Scarani.
\newblock Fluctuation theorems with retrodiction rather than reverse processes.
\newblock \emph{AVS Quantum Science}, 3\penalty0 (4):\penalty0 045601, 2021.
\newblock \doi{https://doi.org/10.1116/5.0060893}.

\bibitem[Crooks(2008)]{crooks2008quantum}
Gavin~E. Crooks.
\newblock Quantum operation time reversal.
\newblock \emph{Physical Review A}, 77\penalty0 (3):\penalty0 034101, 2008.
\newblock \doi{https://doi.org/10.1103/PhysRevA.77.034101}.

\bibitem[Jaynes(2003)]{jaynes2003}
Edwin~T. Jaynes.
\newblock \emph{Probability {{Theory}}: {{The Logic}} of {{Science}}}.
\newblock {Cambridge University Press}, {Cambridge}, 2003.
\newblock ISBN 978-0-521-59271-0.
\newblock \doi{https://doi.org/10.1017/CBO9780511790423}.

\bibitem[Bernardo and Smith(2009)]{bernardo2009}
Jos{\'e}~M. Bernardo and Adrian F.~M. Smith.
\newblock \emph{Bayesian {{Theory}}}.
\newblock {John Wiley \& Sons}, September 2009.
\newblock ISBN 978-0-470-31771-6.
\newblock \doi{https://doi.org/10.1002/9780470316870}.

\bibitem[Asano et~al.(2012)Asano, Basieva, Khrennikov, Ohya, and
  Tanaka]{asano2012}
Masanari Asano, Irina Basieva, Andrei Khrennikov, Masanori Ohya, and Yoshiharu
  Tanaka.
\newblock Quantum-like generalization of the {{Bayesian}} updating scheme for
  objective and subjective mental uncertainties.
\newblock \emph{Journal of Mathematical Psychology}, 3\penalty0 (56):\penalty0
  166--175, 2012.
\newblock \doi{https://doi.org/10.1016/j.jmp.2012.02.003}.

\bibitem[Dezert et~al.(2018)Dezert, Tchamova, and Han]{dezert2018}
Jean Dezert, Albena Tchamova, and Deqiang Han.
\newblock Total {{Belief Theorem}} and {{Generalized Bayes}}' {{Theorem}}.
\newblock In \emph{21st {{International Conference}} on {{Information Fusion}}
  ({{Fusion}} 2018)}, {Cambridge, United Kingdom}, July 2018.
\newblock \doi{https://doi.org/10.23919/ICIF.2018.8455351}.

\bibitem[Parzygnat and Russo(2022)]{parzygnat2020}
Arthur~J. Parzygnat and Benjamin~P. Russo.
\newblock A non-commutative {Bayes'} theorem.
\newblock \emph{Linear Algebra and its Applications}, 644:\penalty0 28--94,
  2022.
\newblock \doi{https://doi.org/10.1016/j.laa.2022.02.030}.

\bibitem[Vanslette(2017)]{vanslette2017}
Kevin Vanslette.
\newblock Entropic {{Updating}} of {{Probabilities}} and {{Density Matrices}}.
\newblock \emph{Entropy}, 19\penalty0 (12):\penalty0 664, December 2017.
\newblock \doi{https://doi.org/10.3390/e19120664}.

\bibitem[Warmuth and Kuzmin(2010)]{warmuth2014}
Manfred~K Warmuth and Dima Kuzmin.
\newblock Bayesian generalized probability calculus for density matrices.
\newblock \emph{Machine learning}, 78\penalty0 (1-2):\penalty0 63, 2010.
\newblock \doi{https://doi.org/10.1007/s10994-009-5133-7}.

\bibitem[Vanslette(2018)]{vanslette2018}
Kevin Vanslette.
\newblock The quantum {{Bayes}} rule and generalizations from the quantum
  maximum entropy method.
\newblock \emph{J. Phys. Commun.}, 2\penalty0 (2):\penalty0 025017, February
  2018.
\newblock \doi{https://doi.org/10.1088/2399-6528/aaaa08}.

\bibitem[Holik et~al.(2014)Holik, Sáenz, and Plastino]{holik2013}
Federico Holik, Manuel Sáenz, and Angel Plastino.
\newblock A discussion on the origin of quantum probabilities.
\newblock \emph{Annals of Physics}, 340\penalty0 (1):\penalty0 293--310, 2014.
\newblock ISSN 0003-4916.
\newblock \doi{https://doi.org/10.1016/j.aop.2013.11.005}.

\bibitem[Fuchs and Schack(2009)]{fuchs2009}
Christopher~A. Fuchs and R{\"u}diger Schack.
\newblock Priors in {{Quantum Bayesian Inference}}.
\newblock \emph{AIP Conference Proceedings}, 1101\penalty0 (1):\penalty0
  255--259, March 2009.
\newblock \doi{https://doi.org/10.1063/1.3109948}.

\bibitem[Giffin and Caticha(2007)]{Giffin07}
Adom Giffin and Ariel Caticha.
\newblock Updating probabilities with data and moments.
\newblock \emph{AIP Conference Proceedings}, 954\penalty0 (1):\penalty0 74--84,
  2007.
\newblock \doi{https://doi.org/10.1063/1.2821302}.

\bibitem[Ali et~al.(2012)Ali, Cafaro, Giffin, Lupo, and Mancini]{Ali12}
Sean~A. Ali, Carlo Cafaro, Adom Giffin, Cosmo Lupo, and Stefano Mancini.
\newblock {On a differential geometric viewpoint of Jaynes' MaxEnt method and
  its quantum extension}.
\newblock \emph{AIP Conference Proceedings}, 1443\penalty0 (1):\penalty0
  120--128, 2012.
\newblock \doi{https://doi.org/10.1063/1.3703628}.

\bibitem[Kostecki(2014)]{kostecki2014}
Ryszard~Paweł Kostecki.
\newblock {L\"uders' and quantum Jeffrey's rules as entropic projections}.
\newblock 2014.
\newblock \doi{https://doi.org/10.48550/arXiv.1408.3502}.

\bibitem[Accardi(1978)]{Accardi78}
Luigi Accardi.
\newblock {Noncommutative Markov Chains Associated to a Preassigned Evolution:
  An Application to the Quantum Theory of Measurement}.
\newblock \emph{Adv. Math.}, 29:\penalty0 226--243, 1978.
\newblock \doi{https://doi.org/10.1016/0001-8708(78)90012-9}.

\bibitem[Accardi and Cecchini(1982)]{Accardi82}
Luigi Accardi and Carlo Cecchini.
\newblock {Conditional expectations in von Neumann algebras and a theorem of
  Takesaki}.
\newblock \emph{Journal of Functional Analysis}, 45\penalty0 (2):\penalty0
  245--273, 1982.
\newblock \doi{https://doi.org/10.1016/0022-1236(82)90022-2}.

\bibitem[Leifer(2006)]{Leifer06}
M.~S. Leifer.
\newblock Quantum dynamics as an analog of conditional probability.
\newblock \emph{Phys. Rev. A}, 74:\penalty0 042310, Oct 2006.
\newblock \doi{https://doi.org/10.1103/PhysRevA.74.042310}.

\bibitem[Coecke and Spekkens(2012)]{coecke12}
Bob Coecke and Robert~W. Spekkens.
\newblock Picturing classical and quantum {{Bayesian}} inference.
\newblock \emph{Synthese}, 186\penalty0 (3):\penalty0 651--696, June 2012.
\newblock \doi{https://doi.org/10.1007/s11229-011-9917-5}.

\bibitem[Ohya and Petz(1993)]{ohya1993}
Masanori Ohya and D{\'{e}}nes Petz.
\newblock \emph{Quantum Entropy and Its Use}.
\newblock Springer Berlin Heidelberg, 1993.
\newblock \doi{https://doi.org/10.1007/978-3-642-57997-4}.

\bibitem[Petz(1986)]{petz1986}
D{\'e}nes Petz.
\newblock Sufficient subalgebras and the relative entropy of states of a von
  {{Neumann}} algebra.
\newblock \emph{Commun.Math. Phys.}, 105\penalty0 (1):\penalty0 123--131, March
  1986.
\newblock \doi{https://doi.org/10.1007/BF01212345}.

\bibitem[Petz(1988)]{Petz:1988usv}
D{\'{e}}nes Petz.
\newblock {Sufficiency of channels over von Neumann algebras}.
\newblock \emph{Quart. J. Math. Oxford Ser.}, 39\penalty0 (1):\penalty0
  97--108, 1988.
\newblock \doi{https://doi.org/10.1093/qmath/39.1.97}.

\bibitem[Petz(2003)]{Petz:2002eql}
D{\'{e}}nes Petz.
\newblock Monotonicity of quantum relative entropy revisited.
\newblock \emph{Rev. Math. Phys.}, 15\penalty0 (01):\penalty0 79--91, 2003.
\newblock \doi{https://doi.org/10.1142/S0129055X03001576}.

\bibitem[Junge et~al.(2018)Junge, Renner, Sutter, Wilde, and Winter]{junge2018}
Marius Junge, Renato Renner, David Sutter, Mark~M. Wilde, and Andreas Winter.
\newblock Universal {{Recovery Maps}} and {{Approximate Sufficiency}} of
  {{Quantum Relative Entropy}}.
\newblock \emph{Ann. Henri Poincar\'e}, 19\penalty0 (10):\penalty0 2955--2978,
  October 2018.
\newblock \doi{https://doi.org/10.1007/s00023-018-0716-0}.

\bibitem[Jurkat and Ryser(1967)]{Jurkat1967}
Wolfgang Jurkat and Herbert~John Ryser.
\newblock Term ranks and permanents of nonnegative matrices.
\newblock \emph{Journal of Algebra}, 5:\penalty0 342--357, 1967.
\newblock \doi{https://doi.org/10.1016/0021-8693(67)90044-0}.

\bibitem[Temme et~al.(2010)Temme, Kastoryano, Ruskai, Wolf, and
  Verstraete]{temmeH2divergenceMixingTimes2010}
Kristan Temme, Michael~J. Kastoryano, M.~B. Ruskai, M.~M. Wolf, and
  F.~Verstraete.
\newblock The {$\chi^2$}-divergence and mixing times of quantum {{Markov}}
  processes.
\newblock \emph{Journal of Mathematical Physics}, 51\penalty0 (12):\penalty0
  122201, December 2010.
\newblock \doi{https://doi.org/10.1063/1.3511335}.

\bibitem[Vandenberghe et~al.(1998)Vandenberghe, Boyd, and
  Wu]{vandenbergheDeterminantMaximizationLinear1998}
Lieven Vandenberghe, Stephen Boyd, and Shao-Po Wu.
\newblock Determinant {{Maximization}} with {{Linear Matrix Inequality
  Constraints}}.
\newblock \emph{SIAM Journal on Matrix Analysis and Applications}, 19\penalty0
  (2):\penalty0 499--533, April 1998.
\newblock \doi{https://doi.org/10.1137/S0895479896303430}.

\bibitem[Grone et~al.(1984)Grone, Johnson, S{\'a}, and Wolkowicz]{Grone84}
Robert Grone, Charles~R. Johnson, Eduardo~M. S{\'a}, and Henry Wolkowicz.
\newblock Positive definite completions of partial {Hermitian} matrices.
\newblock \emph{Linear Algebra and its Applications}, 58:\penalty0 109--124,
  1984.
\newblock \doi{https://doi.org/10.1016/0024-3795(84)90207-6}.

\bibitem[Choi(1975)]{Choi1975}
Man-Duen Choi.
\newblock Completely positive linear maps on complex matrices.
\newblock \emph{Linear Algebra and its Applications}, 10\penalty0 (3):\penalty0
  285--290, 1975.
\newblock \doi{https://doi.org/10.1016/0024-3795(75)90075-0}.

\bibitem[Rudolph(2004)]{rudolphExtremalQuantumStates2004}
Oliver Rudolph.
\newblock On extremal quantum states of composite systems with fixed marginals.
\newblock \emph{Journal of Mathematical Physics}, 45\penalty0 (11):\penalty0
  4035, October 2004.
\newblock \doi{https://doi.org/10.1063/1.1776642}.

\bibitem[Fagnola and Umanit{\`a}(2010)]{fagnolaGeneratorsKMSSymmetric2010}
Franco Fagnola and Veronica Umanit{\`a}.
\newblock {Generators of KMS Symmetric Markov Semigroups on {$\mathcal{B}({\rm
  h})$} Symmetry and Quantum Detailed Balance}.
\newblock \emph{Communications in Mathematical Physics}, 298\penalty0
  (2):\penalty0 523--547, September 2010.
\newblock \doi{https://doi.org/10.1007/s00220-010-1011-1}.

\bibitem[Wolf and Cirac(2008)]{wolf2008dividing}
Michael~M Wolf and J~Ignacio Cirac.
\newblock Dividing quantum channels.
\newblock \emph{Communications in Mathematical Physics}, 279\penalty0
  (1):\penalty0 147--168, 2008.
\newblock \doi{https://doi.org/10.1007/s00220-008-0411-y}.

\bibitem[Leifer and Spekkens(2013)]{Leifer2013}
M.~S. Leifer and Robert~W. Spekkens.
\newblock Towards a formulation of quantum theory as a causally neutral theory
  of {B}ayesian inference.
\newblock \emph{Phys. Rev. A}, 88:\penalty0 052130, Nov 2013.
\newblock \doi{https://doi.org/10.1103/PhysRevA.88.052130}.

\bibitem[Lorenzo et~al.(2013)Lorenzo, Plastina, and
  Paternostro]{paternostro2013}
Salvatore Lorenzo, Francesco Plastina, and Mauro Paternostro.
\newblock Geometrical characterization of non-{Markovianity}.
\newblock \emph{Phys. Rev. A}, 88:\penalty0 020102, Aug 2013.
\newblock \doi{https://doi.org/10.1103/PhysRevA.88.020102}.

\bibitem[Buscemi and Dall’Arno(2019)]{Buscemi2019}
Francesco Buscemi and Michele Dall’Arno.
\newblock Data-driven inference of physical devices: theory and implementation.
\newblock \emph{New Journal of Physics}, 21\penalty0 (11):\penalty0 113029,
  2019.
\newblock \doi{https://doi.org/10.1088/1367-2630/ab5003}.

\bibitem[{Beth Ruskai} et~al.(2002){Beth Ruskai}, Szarek, and
  Werner]{ruskai2002}
Mary {Beth Ruskai}, Stanislaw Szarek, and Elisabeth Werner.
\newblock An analysis of completely-positive trace-preserving maps on m2.
\newblock \emph{Linear Algebra and its Applications}, 347\penalty0
  (1):\penalty0 159--187, 2002.
\newblock ISSN 0024-3795.
\newblock \doi{https://doi.org/10.1016/S0024-3795(01)00547-X}.

\bibitem[Finetti(1974)]{finetti1974}
Bruno~De Finetti.
\newblock \emph{Theory of {{Probability}}: {{A Critical Introductory
  Treatment}}}.
\newblock {Wiley}, 1974.
\newblock ISBN 978-0-471-20141-0.
\newblock \doi{https://doi.org/10.1002/9781119286387}.

\bibitem[Culbertson and Sturtz(2014)]{culbertson14}
Jared Culbertson and Kirk Sturtz.
\newblock A {{Categorical Foundation}} for {{Bayesian Probability}}.
\newblock \emph{Appl Categor Struct}, 22\penalty0 (4):\penalty0 647--662,
  August 2014.
\newblock \doi{https://doi.org/10.1007/s10485-013-9324-9}.

\bibitem[Petz and Ghinea(2011)]{petzIntroductionQuantumFisher2011}
Denes Petz and Catalin Ghinea.
\newblock Introduction to quantum {{Fisher}} information.
\newblock \emph{Quantum Probability and Related Topics}, pages 261--281,
  January 2011.
\newblock \doi{https://doi.org/10.1142/9789814338745_0015}.

\bibitem[Scandi et~al.()Scandi, Abiuso, De~Santis, and Surace]{tim22b}
Matteo Scandi, Paolo Abiuso, Dario De~Santis, and Jacopo Surace.
\newblock Quantum fisher information and its dynamical nature.
\newblock \emph{in preparation}.

\bibitem[Buscemi et~al.(2020)Buscemi, Fujiwara, Mitsui, and
  Rotondo]{BuscemiRotondo2020}
Francesco Buscemi, Daichi Fujiwara, Naoki Mitsui, and Marcello Rotondo.
\newblock Thermodynamic reverse bounds for general open quantum processes.
\newblock \emph{Phys. Rev. A}, 102:\penalty0 032210, Sep 2020.
\newblock \doi{https://doi.org/10.1103/PhysRevA.102.032210}.

\bibitem[Abiuso et~al.(2022)Abiuso, Scandi, Surace, and De~Santis]{tim22a}
Paolo Abiuso, Matteo Scandi, Jacopo Surace, and Dario De~Santis.
\newblock {Characterizing (non-) Markovianity through Fisher Information}.
\newblock 2022.
\newblock \doi{https://doi.org/10.48550/arXiv.2204.04072}.

\bibitem[Csisz{\'a}r et~al.(2004)Csisz{\'a}r, Shields,
  et~al.]{csiszar2004information}
Imre Csisz{\'a}r, Paul~C Shields, et~al.
\newblock Information theory and statistics: A tutorial.
\newblock \emph{Foundations and Trends in Communications and Information
  Theory}, 1\penalty0 (4):\penalty0 417--528, 2004.
\newblock \doi{http://doi.org/10.1561/0100000004}.

\bibitem[Lesniewski and Ruskai(1999)]{lesniewskiMonotoneRiemannianMetrics1999}
Andrew Lesniewski and Mary~Beth Ruskai.
\newblock Monotone {{Riemannian Metrics}} and {{Relative Entropy}} on
  {{Non}}-{{Commutative Probability Spaces}}.
\newblock \emph{Journal of Mathematical Physics}, 40\penalty0 (11):\penalty0
  5702--5724, November 1999.
\newblock ISSN 0022-2488, 1089-7658.
\newblock \doi{https://doi.org/10.1063/1.533053}.

\bibitem[Breuer et~al.(2002)Breuer, Petruccione, et~al.]{breuer2002theory}
Heinz-Peter Breuer, Francesco Petruccione, et~al.
\newblock \emph{The theory of open quantum systems}.
\newblock Oxford University Press on Demand, 2002.
\newblock URL \url{https://doi.org/10.1093/acprof:oso/9780199213900.001.0001}.

\bibitem[Klyachko(2004)]{Klyachko04}
Alexander Klyachko.
\newblock Quantum marginal problem and representations of the symmetric group.
\newblock 2004.
\newblock \doi{https://doi.org/10.48550/arXiv.quant-ph/0409113}.

\bibitem[Bravyi(2004)]{Bravyi04}
Sergey Bravyi.
\newblock Compatibility between local and multipartite states.
\newblock \emph{Quantum Information and Computation}, 4:\penalty0 012--026,
  2004.
\newblock \doi{https://doi.org/10.26421/QIC4.1-2}.

\end{thebibliography}
	
	\appendix

	\section{Estimation of the relative entropy}\label{app:relativeEntropyEstimate}
	We provide here a derivation of Eq.~\eqref{eq:dEstimate}. We want to bound the average relative entropy between the original distribution and the evolved one:
	\begin{align}\label{eq:averageRelativeEntropy2}
		\int_{\mathcal{S}} \de \rho\; D( \rho || \tilde\Phi\Phi ( \rho)) = \int_{\mathcal{S}} \de \rho\; \rho\cdot (\log \rho - \log \tilde\Phi\Phi ( \rho)),
	\end{align}
	where we indicate by $\mathcal{S}$ the space of states. To this end we study separately the two components of the integral in Eq.~\eqref{eq:averageRelativeEntropy2}. First, notice that the first term
	\begin{align}
		\int_{\mathcal{S}} \de \rho \; \rho\cdot \log \rho  = - \langle S(\rho)\rangle ,
	\end{align}
	does not depend on the recovery map. We indicate by $\langle S(\rho)\rangle$ the average Shannon entropy evaluated over the whole space of distributions. Since the quantity in Eq.~\eqref{eq:averageRelativeEntropy2} is non-negative we have the inequality:
	\begin{align}\label{eq:secondPart}
		-\int_{\mathcal{S}} \de \rho \; \rho\cdot \log \tilde\Phi\Phi ( \rho) \geq \langle S(\rho)\rangle,
	\end{align}
	with equality only for perfect retrieval, i.e., $\tilde\Phi\Phi \equiv \id$. Notice that the minimisation of Eq.~\eqref{eq:averageRelativeEntropy2} is equivalent to the minimisation of this last integral. We can also majorize this term as:
	\begin{align}
		&-\int_{\mathcal{S}} \de \rho \; \rho\cdot \log \tilde\Phi\Phi ( \rho) \leq -\int_{\mathcal{S}} \de \rho \; \sum_i  \log (\tilde\Phi\Phi (\rho))_i=\\
		&\qquad= -\frac{1}{|\det \tilde{\Phi}\Phi|}\int_{\tilde\Phi\Phi (\mathcal{S})} \de \sigma \; \sum_i  \log  \sigma_i\leq\\
		&\qquad\leq -\frac{1}{|\det \tilde{\Phi}\Phi|}\int_{\mathcal{S}} \de \sigma \; \sum_i  \log  \sigma_i = \frac{K}{|\det \tilde{\Phi}\Phi|},
	\end{align}
	where in the first inequality we used the fact that ${\rho_i\leq 1}$ for every $i$, in the second line we changed variables, in the third we extended the integral from the image of $\tilde\Phi\Phi $ to the whole space, and lastly we implicitly defined the positive constant $K$ to be the integral in the last line.
	
	Hence, we have that Eq.~\eqref{eq:secondPart} can be estimated as:
	\begin{align}
		\langle S(\rho)\rangle\leq	-\int_{\mathcal{S}} \de \rho\; \rho\cdot \log \tilde\Phi\Phi ( \rho) \leq\frac{K}{|\det \tilde{\Phi}\Phi|}.
	\end{align}
	Subtracting to every term the average Shannon entropy gives the result in Eq.~\eqref{eq:dEstimate}. These inequalities give a rough idea of why minimising the determinant helps minimising the average discrepancy between the initial state and the retrieved one.
	
	%\ms{
	\section{Additional bounds on the quality of retrieval}\label{app:relEntropyEst2}
	
	In order to present the following bounds at the right level of generality we introduce the Csizár contrast functions between two classical distributions, which reads~\cite{csiszar2004information}:
	\begin{align}
		H_g(\rho||\sigma) : = \sum_i \, \rho_i \,g\norbra{\frac{\sigma_i}{\rho_i}},
	\end{align}
	where $g$ is an arbitrary positive convex function such that $g(0)=0$. The relative entropy is part of this family, which also contains many other well-known quantifiers of statistical distance such as the Hellinger distance or the $\chi^2$-divergence. Interestingly, for close-by states, all the contrast functions collapse into the same quantity, called Fisher information:
	\begin{align}
		H_g(\rho||\rho+\delta\rho) \simeq \frac{1}{2} \Tr{\delta \rho\,\cJ_{\rho}^{-1}[\delta \rho]} = \frac{1}{2} \sum_i \frac{\delta\rho_i^2}{\rho_i},
	\end{align}
	where $|\delta\rho|\ll 1$, and we used the matrix form of the equation to keep the analogy with the quantum case.
	Since we won't use any specific property of these $g$s, the reader unfamiliar with these concepts can substitute in all the classical calculations $H_g(\rho||\sigma)$ with the relative entropy. 
	
	The generalisation to the quantum regime of $H_g(\rho||\sigma)$ was given in~\cite{lesniewskiMonotoneRiemannianMetrics1999}, to which we refer for the general form. The only quantum contrast function we will consider here is given by:
	\begin{align}\label{eq:qContrastFunction}
		H_{{\rm sq}}(\rho||\sigma) = \Tr{\sqrt{\rho}(\rho-\sigma)\sqrt{\sigma^{-1}}},
	\end{align}
	corresponding to the matrix convex function $g(x)=\sqrt{x^{-1}}-\sqrt{x} $ (where we use the notation from~\cite{lesniewskiMonotoneRiemannianMetrics1999}). It is also useful to introduce the symmetrised version of this contrast function:
	\begin{align}
		H_{{\rm sq}}^{{\rm symm}}(\rho||\sigma) &= \frac{1}{2}\norbra{H_{{\rm sq}}(\rho||\sigma)+H_{{\rm sq}}(\sigma||\rho)}=\\
		&=\frac{1}{2}\Tr{\sqrt{\rho^{-1}}(\rho-\sigma)\sqrt{\sigma^{-1}}(\rho-\sigma)}\,.\label{eq:ultima}
	\end{align}
	 It is a known fact that quantum contrast functions locally behave in the same way as their symmetrised version~\cite{lesniewskiMonotoneRiemannianMetrics1999}. Hence, we can use Eq.~\eqref{eq:ultima} to find the expansion:
	\begin{align}
		H_{{\rm sq}}(\rho||\rho+\delta\rho)\simeq	H_{{\rm sq}}^{{\rm symm}}(\rho||\rho+\delta\rho) \simeq \frac{1}{2} \Tr{\delta \rho\,\J_{\rho}^{-1}[\delta \rho]},\label{eq:expansionFisher}
	\end{align}
	where $\Tr{|\delta \rho|}\ll1$.
	It should be noticed that the assumption~(\ref{ax:q4}) makes $\tilde\Phi\Phi$ self-adjoint with respect to the scalar product induced by $\J_\pi$. In fact, one has:
	\begin{align}
		&\Tr{A\,\J_\pi^{-1}[\tilde\Phi\Phi(B)]} = \Tr{\J_\pi[(\tilde\Phi\Phi)^\dagger(\J_\pi^{-1}[A])]\,\J_\pi^{-1}[B]}=\nonumber\\
		&=\Tr{\tilde\Phi\Phi(A)\,\J_\pi^{-1}[B]},
	\end{align}
	where the calculations can be easily verified by using the explicit expression of $\J_\pi[\rho] =\sqrt{\pi}\rho\sqrt{\pi}$. The same calculations carry out to the classical case by substituting $\J_\pi$ with $\cJ_\pi$ (in fact, this holds for all the calculations in this section). 
	
	Thanks to the self-adjointness and adjoint-preserving properties of $\tilde\Phi\Phi$ we can find an orthonormal basis of self-adjoint operators $\{E_i\}$  such that $\tilde\Phi\Phi[E_i] := \varphi_i \,E_i$ and $\Tr{E_i\,\J_\pi^{-1}[E_j]} = \delta_{i,j}$. Moreover, due to principle~(\ref{ax:q5}), all the $\varphi_i$ are positive and less than one (since $\tilde\Phi\Phi$ is a CP-map) meaning that we can express them as $\varphi_i = e^{-\lambda_i}$, where $\lambda_i$ are all positive. Notice that the $\lambda$s are also connected to the determinant of the forth-and-back map by the relation:
	\begin{align}
		\log \det\tilde\Phi\Phi  = \Tr{\log\tilde\Phi\Phi } = - \sum_i \lambda_i\,.
	\end{align}
	Consider now the contraction rate close to a state $\rho$:
	\begin{align}
		\eta^F_g(\rho,\delta\rho) :=\frac{H_g(\rho||	\rho+\delta\rho)-H_g(\tilde\Phi\Phi(\rho)||\tilde\Phi\Phi(\rho+\delta\rho))}{H_g(\rho||	\rho+\delta\rho)} ,\label{eq:52}
	\end{align}
	where $H_g(\rho||\sigma)$ is a generic contrast function in the classical case, or the one in Eq.~\eqref{eq:qContrastFunction} in the quantum one. Then, thanks to the expansion in Eq.~\eqref{eq:expansionFisher} and the self-adjointness of $\tilde\Phi \Phi$, we can rewrite the contraction rate close to the prior $\pi$ as:
	\begin{align}
		\eta^F_g(\pi,\delta\rho) &\simeq\frac{\Tr{\delta\rho\,\J_\pi^{-1}[\delta\rho]} - \Tr{\delta\rho\,\J_\pi^{-1}[(\tilde\Phi\Phi)^2(\delta\rho)]} }{\Tr{\delta\rho\,\J_\pi^{-1}[\delta\rho]}} =\\
		&=\frac{\sum_i |\delta\rho_i|^2 (1-e^{-2\lambda_i})}{\sum_i |\delta\rho_i|^2},
	\end{align}
	where $\delta\rho_i:= \Tr{\delta\rho\,\J_\pi^{-1}[E_i]}$ are the components of $\delta\rho$ in the eigenbasis of $\tilde\Phi\Phi$.
	
	Then, thanks to the inequality $1-e^{-x}\leq x$ holding for positive $x$ (with equality only for $x\equiv 0$), we can bound the contraction rate by:
	\begin{align}
		\eta^F_g(\pi,\delta\rho)&\leq 2\,\frac{\sum_i |\delta\rho_i|^2 \lambda_i}{\sum_i |\delta\rho_i|^2}\leq\\
		&\leq2\sum_i \lambda_i = 2 \log\det(\tilde\Phi\Phi)^{-1} \,.
	\end{align}
	
	We are now ready to prove the first result of this Appendix. To this end, we introduce contraction coefficients akin to the one used in~\cite{lesniewskiMonotoneRiemannianMetrics1999}:
	\begin{align}
		\eta_g^{RE} &:= \inf_{\rho,\sigma}\, \frac{H_g(\rho||	\sigma)-H_g(\tilde\Phi\Phi(\rho)||\tilde\Phi\Phi(\sigma))}{H_g(\rho||	\sigma)}\,;\\
		\eta_g^{F} &:= \inf_{\rho,\delta\rho}\, \eta^F_g(\rho,\delta\rho)\,.
	\end{align}
	In there, it was shown that $\eta_g^{RE}\leq\eta_g^{F}$.  Hence, we have the following chain of inequalities
	\begin{align}
		\eta_g^{RE}\leq\eta_g^{F}\leq \eta^F_g (\pi,\delta\rho)\leq 2 \log\det(\tilde\Phi\Phi)^{-1} ,
	\end{align} 
	where the second inequality follows from the fact that the infimum is always smaller than the value of the function at a specific point. This 
	proves Eq.~\eqref{eq:20} for general contrast functions and so, in particular, also for the relative entropy.
	
	The proof of Eq.~\eqref{eq:19} is completely analogous: rewrite the quantity in consideration as:
	\begin{align}
			H_g(\pi +\delta\rho ||&\tilde\Phi\Phi(\pi +\delta\rho )) \simeq\\
			\simeq& \frac{1}{2}\Tr{\delta\rho\,\J_\pi^{-1}[(\id-\tilde\Phi\Phi)^2(\delta\rho)]}=\\
			&=\frac{1}{2}\sum_i |\delta\rho_i|^2 (1-e^{-\lambda_i})^2.
	\end{align}
	Then, thanks to the inequality $(1-e^{-x})^2\leq x/2$ holding for positive $x$, we directly have:
	\begin{align}
		H_g(\pi +\delta\rho ||\tilde\Phi&\Phi(\pi +\delta\rho )) \leq \frac{1}{4}\sum_i |\delta\rho_i|^2 \lambda_i\leq\\
		&\leq \frac{H_g(\pi  ||\pi +\delta\rho ) }{2} \log\det(\tilde\Phi\Phi)^{-1}.
	\end{align}
	Again specialising to the relative entropy for classical distributions finally gives Eq.~\eqref{eq:19}. 
	
	It is important to keep in mind that these computations only hold in the quantum case for the contrast function in Eq.~\eqref{eq:qContrastFunction}. Still, if principle~(\ref{ax:q4}) gets modified with the requirement that $\tilde\Phi\Phi$ satisfies the canonical definition of detailed balance (i.e., the one given, for example, in~\cite{breuer2002theory}) all the steps can be generalised to any quantum contrast function. This extension is straightforward, but involves a number of technical details outside of the scope of the present publication. For this reason, we defer its treatment to subsequent works~\cite{tim22b}.
	%}
	
	\begin{figure}
		\centering
		\includegraphics[width=0.9\linewidth]{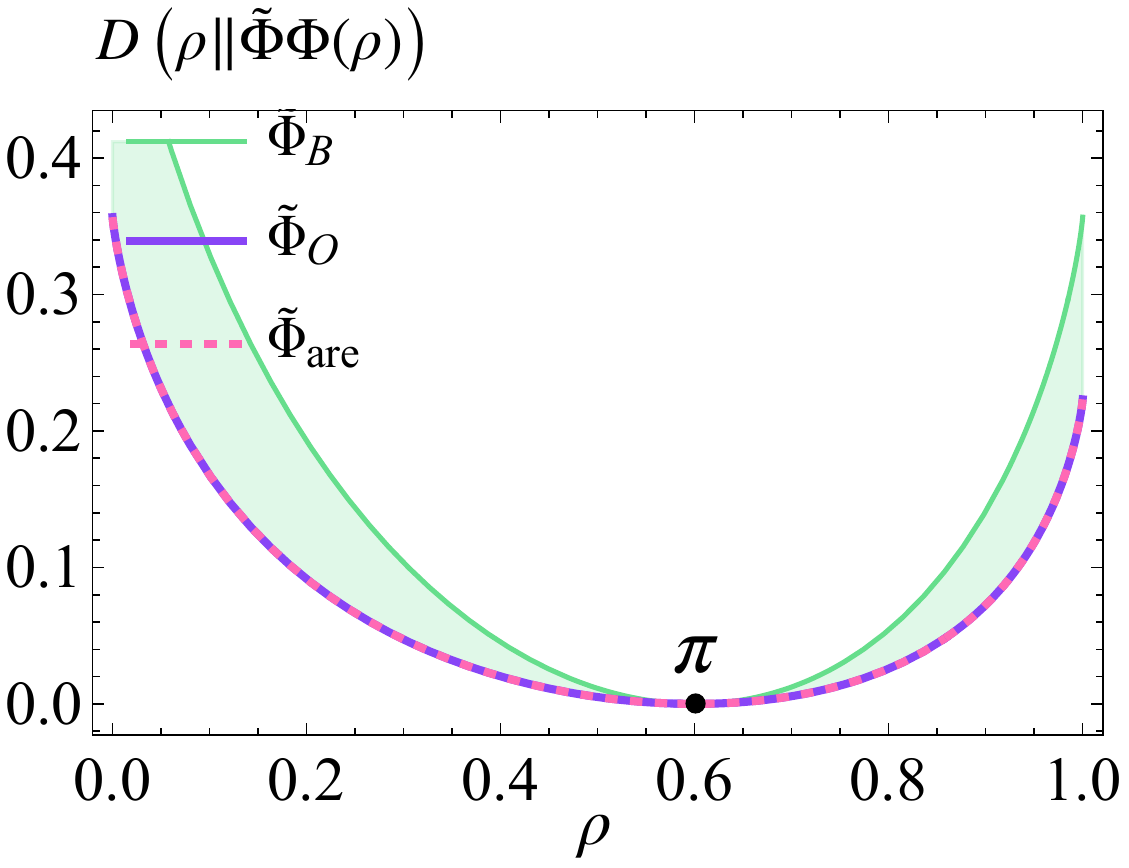}
		\caption{\js{Relative entropy between a distribution and its evolution forwards and backwards. We consider probability vectors $\rho=[\rho,1-\rho]$. As it can be seen, choosing the prior $\pi$ as the fixed point of $\tilde{\Phi}\Phi$, the optimal map $\tilde\Phi_O$ coincides with $\tilde{\Phi}_{are}$ and both outperform the Bayes retrodiction $\tilde\Phi_B$ in retrieving the original distribution in the whole space.}}
		\label{fig:Are-Punto-Fisso}
	\end{figure}
	
	\js{
		\section{Comparison with minimisation of the average relative entropy}
		In section~\ref{sec:optimalStateRetrieval} we used the relative entropy of recovery to evaluate and compare the quality of the retrieval of the single state for the optimal and the Bayes' inspired retrieval map. Given a specific state $\rho$, a map $\Phi$ and a retrieval map $\tilde{\Phi}$, the quality of the retrieved state $\tilde{\Phi}\Phi(\rho)$ is higher the smaller its relative entropy of retrieval $D\left(\rho\|\tilde{\Phi}\Phi(\rho) \right)$ is. In the examples of Figure~\ref{fig:figcomposed} the optimal retrieval map outperforms the Bayes' inspired retrieval map. A natural question in this context is how the optimal map compares with the retrieval map $\tilde{\Phi}_{are}$ obtained by directly minimising the average relative entropy on every input state. In Figure~\ref{fig:Are-Punto-Fisso} and Figure~\ref{fig:Are-Punto-Random} we consider the same example of Figure~\ref{fig:figcomposed} and compute $\tilde{\Phi}_{are}$ as the stochastic map that minimises the relative entropy of recovery on average on every input state
		\begin{equation}
		\tilde{\Phi}_{are}=\argmin_{\substack{\tilde \Phi \,{\rm stochastic}}} \int d\rho D\left(\rho \|\tilde{\Phi}\Phi(\rho) \right).
		\end{equation}
Note that $\tilde{\Phi}_{are}$ does not depend on the choice of the prior $\pi$, but it only depends on the map $\Phi$.
	In Figure~\ref{fig:Are-Punto-Fisso} we see that choosing the prior of the optimal retrieval map as the fixed point of $\tilde{\Phi}_{are}\Phi$, the optimal map $\tilde\Phi_O$ coincides with $\tilde{\Phi}_{are}$.
	}

	\section{Constraints on the Choi state of CPTP maps with a given transition}\label{app:Choi}
	In this section we show how to formulate the constraints in Eq.~\eqref{eq:TPConstraint} and Eq.~\eqref{eq:ax3Constraint} in terms of the Choi state of $\Lambda^\Psi$.
	
	\begin{figure}
		\centering
		\includegraphics[width=0.9\linewidth]{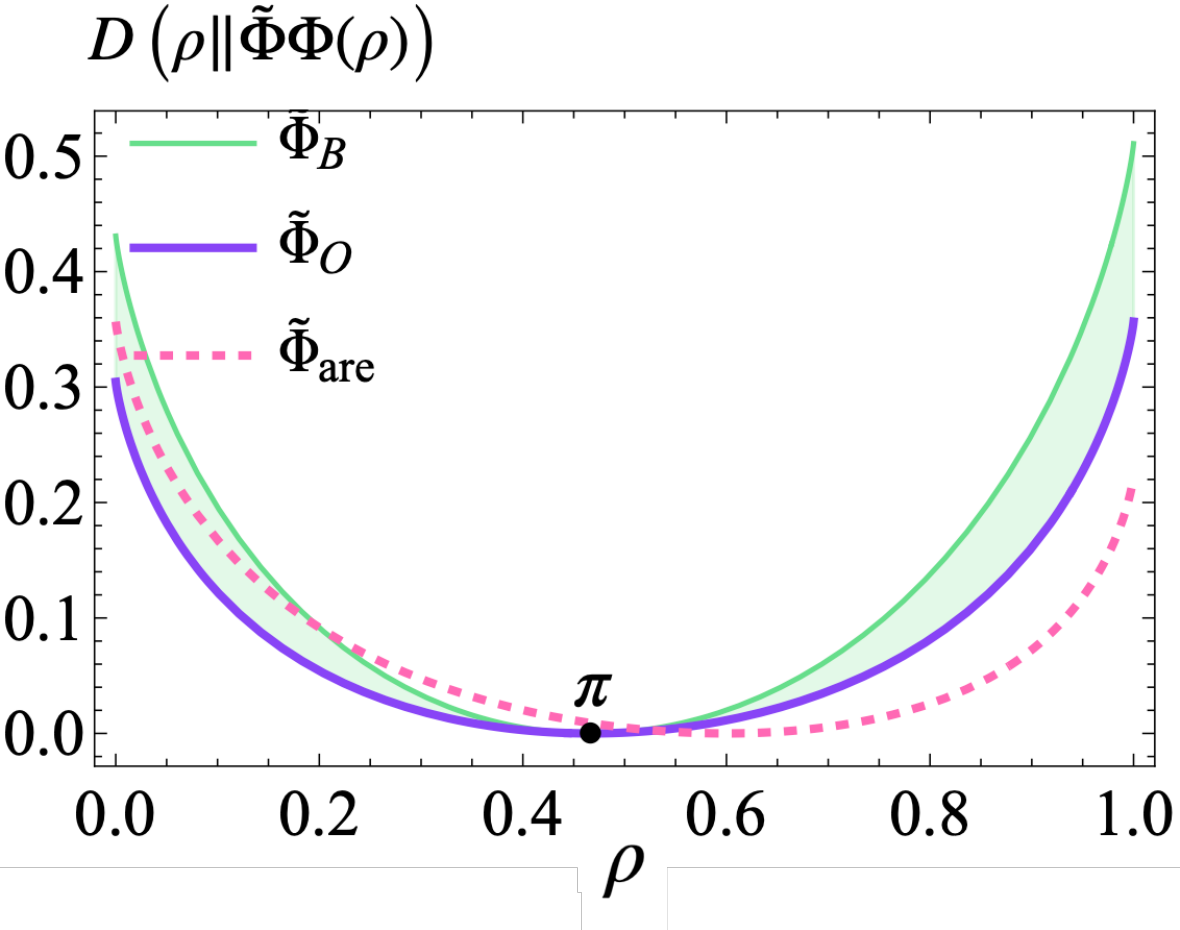}
		\caption{\js{Relative entropy between a distribution and its evolution forwards and backwards. We consider probability vectors $\rho=[\rho,1-\rho]$. As it can be seen, choosing the prior $\pi$ as a different point with respect to the fixed point of $\tilde{\Phi}_{are}\Phi$, the effect of the optimal map $\tilde\Phi_O$ differs from the one of $\tilde{\Phi}_{are}$.}}
		\label{fig:Are-Punto-Random}
	\end{figure}
	
	Consider the maximally entangled state:
	\begin{align}
		\ket{\Omega} := \frac{1}{\sqrt{d}}\sum_{i=0}^{d-1} \; \ket{i}_A\otimes\ket{i}_B.
	\end{align}
	The unnormalised Choi state of the map $\Psi$ is defined by the formula $\mathcal{C}^{\Psi}:=d\,{(\ids_A\otimes\Psi) [\ket{\Omega}\bra{\Omega}]}$. Then, the application of $\Psi$ to a state $\rho$ can be equivalently expressed  as $\Psi[\rho] = \tr_A\sqrbra{(\rho^T\otimes\ids)\,\mathcal{C}^{\Psi}}$. Moreover, it also holds that $\Psi^\dagger[\rho] = (\tr_B\sqrbra{(\ids\otimes\rho)\,\mathcal{C}^{\Psi}})^T$. Finally, the Choi--Jamiołkowski isomorphism states that a map $\Psi$ is completely positive if and only if the Choi state $\mathcal{C}^\Psi$ is positive definite.
	
	We can pass to characterise $\mathcal{U}_Q(\sigma,\pi)$ in terms of the corresponding Choi states. Consider a map $\Lambda\in\mathcal{U}_Q(\sigma,\pi)$. Eq.~\eqref{eq:TPConstraint} translates to:
	\begin{align}
		\tr_B\sqrbra{\mathcal{C}^{\Lambda}} &= \tr_B\sqrbra{(\ids\otimes\ids)\mathcal{C}^{\Lambda}}= (\Lambda^\dagger[\ids])^T=(\pi)^T.
	\end{align}
	Moreover, it is also follows that Eq.~\eqref{eq:ax3Constraint} translates to:
	\begin{align}
		\tr_A\sqrbra{\mathcal{C}^\Lambda} =\tr_A\sqrbra{(\ids\otimes\ids)\mathcal{C}^{\Lambda}} =  (\Lambda)[\ids] = \sigma.
	\end{align}
	Since $\mathcal{C}^\Lambda$ is positive semidefinite and $\tr[\mathcal{C}^\Lambda]=1$, the set $\mathcal{U}_Q(\sigma,\pi)$, thanks to the Choi–Jamiołkowski isomorphism, is isomorphic to the set of all the bipartite quantum states $\rho_{AB}$ compatible with the two marginals $\rho_A=\sigma$ and $\rho_B=(\pi)^T$.
	
	This identification allows to constrain the spectrum of the Choi states in $\mathcal{U}_Q(\sigma,\pi)$. In fact, one can construct a system of linear inequalities depending on the spectrum of  $\rho_A$ and $\rho_B$ to constrain the spectrum of $\rho_{AB}$~\cite{Klyachko04,Bravyi04}. Moreover, similarly with what happened for the classical case, one can use the spectrum of $\rho_{AB}$ to associate to a set of scalars a map in $\mathcal{U}_Q(\sigma,\pi)$. Differently from the classical case, though, this association is not unique: in fact, the symmetry in Eq.~\eqref{eq:unit-invariance} preserves the spectrum of the Choi matrix, so we can only associate to each set of scalars a unique equivalence class, but not a unique map.

	%For each compatible spectrum $\vec{\lambda}_{AB}$ we select a representative matrix $\mathcal{C}^{[\vec{\lambda}_{AB}]}$ satisfying \eqref{eq:reqPSDChoi}\eqref{eq:reqTr1Choi}\eqref{eq:reqTr2Choi}. Each element $\mathcal{C}$ of $\mathcal{U}_Q(\rho,\sigma)$ is individuated by the triplet $(\vec{\lambda}_{AB},U_{A},U_{B})$ as $\mathcal{C}=(U_{A}\otimes U_{B})\mathcal{C}^{\vec{\lambda}_{AB}}(U_{A}^{\dagger}\otimes U_{B}^{\dagger})$. We have then that for each triplet $(\vec{\lambda}_{AB},U_{A},U_{B})$:
	%\begin{align}
	%\Lambda_{\rho \atop \sigma}[X] & =\tr_2\left[(\id_1\otimes X^T)(U_{A}\otimes U_{B})\mathcal{C}^{\vec{\lambda}_{AB}}(U_{A}^{\dagger}\otimes U_{B}^{\dagger}) \right] = \nonumber \\
	%& = U_{A}\tr_2\left[(\id_1\otimes U_{B}X^TU_{B}^{\dagger})\mathcal{C}^{\vec{\lambda}_{AB}} \right]U_{A}^{\dagger} := \nonumber \\
	%& = U_{A} \Lambda^{[\vec{\lambda}_{AB}]} \left[U_{B}X^TU_{B}^{\dagger}\right] U_{A}^{\dagger}
	%\end{align}

\end{document}